\def\inserfig#1#2{\centerline{\includegraphics[width=#1\linewidth]{#2}}}

\documentclass[11pt]{article}

\usepackage[dvips]{graphicx}
\usepackage{epsfig}

\usepackage{fullpage}

\usepackage{amssymb,amsmath,graphics,color}

\usepackage{epsfig}

\newenvironment{proof}{\noindent\textit{Proof: }}{{\hfill $\Box$}}
\newtheorem{lemma}{Lemma}[section]
\newtheorem{theorem}[lemma]{Theorem}
\newtheorem{proposition}[lemma]{Proposition}
\newtheorem{corollary}[lemma]{Corollary}

\newtheorem{definition}[lemma]{Definition}
\newtheorem{remark}[lemma]{Remark}

\def\G{Gr}
\begin{document}

\title{\textbf{Split Decomposition and Graph-Labelled Trees: \\Characterizations and Fully-Dynamic Algorithms \\ for Totally Decomposable Graphs}\thanks{Work supported by the French research grant ANR-06-BLAN-0148-01 ``\textit{Graph Decompositions and Algorithms} - \textsc{graal}".
This paper completes and develops the extended abstract \cite{GP07}.
}}

\author{Emeric Gioan \and Christophe Paul\thanks{Research partially conducted while C. Paul was on Sabbatical at School of Computer Science, McGill University, Montr\'eal, Canada}}

\date{CNRS - LIRMM, Universit\'e de Montpellier 2, France \\~\\ \today}

\maketitle

%------------------------------------------------------------------------------------------------------------------
%------------------------------------------------------------------------------------------------------------------
\begin{abstract}
In this paper, we revisit the split decomposition of graphs and give new combinatorial and algorithmic results for the class of totally decomposable graphs, also known as the \emph{distance hereditary graphs}, and for two non-trivial subclasses, namely the \emph{cographs} and the \emph{$3$-leaf power graphs}. Precisely, we give strutural and incremental characterizations, leading to optimal fully-dynamic recognition algorithms for vertex and edge modifications, for each of these classes.
These results rely on the new combinatorial framework of \emph{graph-labelled trees} used to represent the split decomposition of general graphs. %These results rely on the new combinatorial framework of \emph{graph-labelled trees} used to represent the split decomposition. 
%%%%%which also captures the modular decomposition of graphs and thereby unify these two decomposition techniques.
The point of the paper is to use bijections between the aforementioned %these 
graph classes and graph-labelled trees whose nodes are labelled by cliques and stars. 
%EME-19-12-10 modif de la phrase, mais je suis pour laisser ca dans l'abstract
%Doing so, we also derive an intersection model for distance hereditary graphs, which answers an open problem.
We mention that this bijective viewpoint yields directly an intersection model for the class of distance hereditary graphs.
\end{abstract}

%------------------------------------------------------------------------------------------------------------------
%------------------------------------------------------------------------------------------------------------------
\section{Introduction}

The \emph{$1$-join composition} and its complementary operation, the \emph{split decomposition}, range among the classical operations in graph theory. It was introduced by Cunningham and Edmonds~\cite{Cun82,CE80} in the early 80's and has, since then, been used in various contexts such as perfect graph theory~\cite{Hsu87}, circle graphs~\cite{Bou87}, clique-with~\cite{Cou06} or rank-width~\cite{Oum05b}. The first polynomial time algorithm to compute the split decomposition of a graph, proposed in~\cite{Cun82}, runs $O(n^3)$ time complexity. It was later improved by Ma and Spinrad~\cite{SM94} who described an $O(n^2)$ time algorithm. So far Dahlhaus' linear time algorithm~\cite{Dah00} is the fastest. Also, we mention the recent work \cite{CHL08} which nicely reformulates underlying routines from \cite{Dah00}.

Roughly speaking, a split is a bipartition of the vertices of a graph satisfying certain properties (see Definition~\ref{def:split}). Computing the split decomposition of a graph consists in recursively decompose that graph according to bipartitions that are splits. This process naturally yields a (split) decomposition tree~\cite{Cun82,CE80} which represents the used bipartitions. However such a tree does not keep track of the adjacency of the input graph. Thereby alternative representations of the split decomposition have been proposed. So far, the \emph{split decomposition graph} appearing in~\cite{CD99,Lan01,GP03,Cou06} seems to be the most commonly used representation.
As an example of another related representation, let us mention the $\Delta$-confluent graphs used for
distance hereditary graph drawing~\cite{EGY05}.

This paper starts with an adaptation of the split decomposition graph into a new and simple combinatorial structure, namely \emph{graph-labelled trees}. A \emph{graph-labelled tree} is a tree in which every internal node $u$  is labelled by a graph $G_u$ whose vertices, called \emph{marker-vertices}, are in one-to-one correspondence with the tree-edges incident to $u$. The definition of graph-labelled trees is independent of the split decomposition. But equipped with the notion of \emph{accessibility}, it precisely catches the combinatorial structure studied in \cite{Cun82} and provides a representation of the adjacencies of the graph to be decomposed. 
A node or a leaf $u$ is \emph{accessible} from a leaf $l\neq u$ if for every tree-edges $e=wv$ and $e'=vw'$ on the $l,u$-path in $T$, $e$ and $e'$ are mapped to adjacent marker vertices in $G_u$. 
%
%
%Two leaves (or marker-vertices) $x=m_0$ and $y=m_t$ of a graph-labelled tree are \emph{accessible} from each another if, considering the path $P$ of the tree between $x$ and $y$ (or between the internal nodes to which they belong), the marker-vertices mapped to tree-edges of $P$ provide an \emph{alternating path} between $x$ and $y$, \textit{i.e.} a sequence of marker-vertices $m_1\dots m_{t-1}$ such that $m_{i-1}m_i$ ($0< i\leq t$) is either a tree-edge or a  graph-label edge, and such that if $m_{i-1}m_i$ ($0< i<t$) is a tree-edge (resp. graph-label edge), then $m_im_{i+1}$ is a graph-label edge (resp. tree-edge). 
%
Every graph-labelled tree is associated with a graph, its \emph{accessibility graph}, whose vertex set is the leaf set of the tree. Two vertices $x$ and $y$ of the accessibility graph are adjacent if and only if the corresponding leaves are accessible from each another.

Surprisingly, revisiting the split decomposition under this original approach yields new combinatorial and algorithmic results, as well as alternative proofs or simpler constructions of previously known results. Section 2 introduces the combinatorial framework of graph-labelled trees which apply to arbitrary graphs. The main results of split decomposition theory are revisited from the graph-labelled trees viewpoint. The split decomposition can be seen as a refinement of the modular decomposition~\cite{Gal67,HP10}. We then describe links between these two graph decompositions techniques in terms of graph-labelled trees. We also establish useful general lemmas.

\medskip
The rest of the paper concentrates on \emph{totally decomposable} graphs (with respect to the split decomposition), also known as the \emph{distance hereditary graphs}~\cite{BM86,HM90}. Distance hereditary graphs play an important role in other classical decomposition techniques since they are exactly the graphs of rank-width $1$~\cite{Oum05b} and range among the elementary graphs of clique-width $3$~\cite{CHL00}. The family of distance hereditary graphs contains a number of well-studied graph classes such as \emph{cographs} which are the graphs totally decomposable by the modular decomposition and \emph{$3$-leaf powers} which form a subfamily of chordal distance hereditary graphs. We apply our techniques to these latter two graph families. Our results are consequences of 
%split decomposition based 
characterizations of the three graph classes we consider (distance hereditary graphs, cographs and $3$-leaf powers). Each of these characterizations, translated into the graph-labelled tree setting, establishes a one-to-one correspondence between the graph class and a set of clique-star labelled trees\footnote{Clique-star (labelled) trees are graph-labelled trees whose graph-labels are cliques (complete graphs) or stars (complete bipartite graphs $K_{1,t}$).}  that satisfy some simple conditions on the distribution of star and clique labels on its nodes.

Our first result, although not the most important, witnesses the relevance of the graph-labelled tree approach to study the split decomposition. The bijection between the clique-star trees and distance hereditary graphs together with the notion of accessibility naturally yields an intersection model that characterizes distance hereditary graphs (Theorem 3.2). Though it was established that distance hereditary graphs form an intersection graph family~\cite{MM99}, 
%EME-19-12-10 attenuation a la place de "`not yet discovered"'
no intersection model had been explicitely given (see~\cite{Spi}, or \cite{Spi03} page 309).

Among the main contributions of the paper, we develop vertex \emph{incremental characterizations} for distance hereditary graphs, cographs and $3$-leaf powers (see Section 3). That is, for each of these three graph classes, say $\mathcal{F}$, we provide a necessary and sufficient condition under which adding a vertex $x$ adjacent to a certain neighborhood $S$ in a given graph $G\in \mathcal{F}$, yields a graph $G'=G+(x,S)$ which also belongs to $\mathcal{F}$. In comparison, a vertex elimination ordering characterization (see \emph{e.g.}~\cite{BLS99}) only provides sufficient conditions under which a vertex can be added.
%Notice that an incremental characterization can be seen as a strict generalization of a vertex elimination ordering characterization~\cite{BLS99}. 
The incremental characterization of distance hereditary graphs (Theorem \ref{th:charac}) is new. Restricted to cographs (Theorem \ref{th:inc-cograph}), it is equivalent the known incremental characterization of cographs~\cite{CPS85} which is based on modular decomposition. We then derive a new incremental characterization of $3$-leaf powers (Theorem \ref{th:3-leaf}).

We also provide \emph{edge-modification characterizations} (see Section 5): necessary and sufficient conditions under which for a given graph $G$ belonging to a class of graphs $\mathcal{F}$, the addition (or deletion) of an edge $e$ of $G$ results in a graph of $\mathcal{F}$. 
Let us point out that an edge-modification characterization (or algorithm) cannot be used to derive a vertex-incremental characterization (or algorithm), since removing/adding an edge incident to a vertex may lead out of the class while adding/removing all edges adjacent to this vertex may not. Indeed we exhibit an example (Remark \ref{rk:contre-exemple}) of distance hereditary graph (and cograph) containing a vertex $x$ such that removing any edge incident to $x$ results in a non-distance hereditary graph.
An edge-modification characterization was known for distance hereditary graphs~\cite{TC07} and for cographs~\cite{SS04} but not for $3$-leaf powers. 
Our characterization for distance hereditary graphs consists in testing whether the path between the two leaves corresponding to the vertices incident to the modified edge has length at most $4$ and belongs to a small given finite set.
So, unlike the characterization proposed in ~\cite{TC07}, which is based on the global breadth-first search layering structure of distance hereditary graphs~\cite{HM90}, ours is really local, have simpler and shorter proofs and is a natural generalization of the edge-modification characterization of cograph of~\cite{SS04}. 
Our edge-modification characterizations of cographs and $3$-leaf powers are derived from our DH graph one.

\medskip
These characterizations (incremental and edge-modification) are then used to design \emph{fully-dynamic recognition algorithms}. For a class $\mathcal{F}$ of graphs, the task is to maintain a representation of the input graph under vertex and edge modifications as long as the graph belongs to $\mathcal{F}$. Let us point out that the series of modifications is not known in advance. 
In order to ensure locality of the computation, most of the known dynamic graph algorithms are based on decomposition techniques. For example, the SPQR-tree data structure has been introduced in order to dynamically maintain the 3-connected components of a graph which allows on-line planarity testing~\cite{DT96}. Existing literature on this problem includes representation of chordal  graphs~\cite{Iba99}, proper interval graphs~\cite{HSS02}, cographs~\cite{SS04}, directed cographs~\cite{CP04}, permutation graphs~\cite{CP05}. The data structures used for the last four graph families are strongly related to the modular decomposition tree~\cite{Gal67}.

For each of the three aforementioned classes of graphs, we provide an optimal fully-dynamic algorithm that maintains the split tree representation. The time complexity is linear in the number of edges involved in each modifications (\textit{i.e.} number of neighbors in case of vertex modifications). Our main algorithmic result is the vertex-insertion algorithm for distance hereditary graphs (Subsection \ref{sub:vertex-DH}). Briefly, it amounts to: first, a single search of the subtree of the split tree spanned by the neighbors of the new vertex $x$ to locate %in that subtree 
where the new leaf $x$ should be inserted (if possible); and then, a simple local transformation of the graph-labelled tree. As distance hereditary graphs form an hereditary class, the vertex-deletion routine consists of an easy local transformation. When adapted to cographs, our vertex-only dynamic algorithm (Subsection \ref{sub:vertex-cograph}) is equivalent to the one of \cite{CPS85}. No such algorithm was known for 3-leaf powers (Subsection \ref{sub:vertex-3leaf}). The edge-only dynamic algorithms are direct consequences of the edge-modification characterizations.

Finally, let us observe that as distance hereditary graphs, cographs and $3$-leaf power graphs are hereditary graph families, our fully dynamic recognition algorithms can be used in the context of static graphs as well. This yields, for each of the three graph classes, linear time recognition algorithms (Corollary \ref{cor:static}) to be compared with previous ones (\cite{HM90,DHP01,Bre05} and \cite{NUU07} for distance hereditary graphs). Moreover, our bijective representations allow to derive directly easy isomorphism tests for elements of these classes (Corollary \ref{cor:iso}).

The algorithmic results presented in this paper are summarized in the table below.

\begin{table}[h]
	\centering
		\begin{tabular}{|l|l|l|r|}
		\hline
			distance hereditary & vertex-only & Subsections \ref{sub:vertex-DH} and \ref{sub:vertex-del-DH} & new \\
			\cline{2-4}
			 graphs & edge-only & Subsection \ref{sub:edge-DH} & independent of and shorter than \cite{TC07} \\
		\hline
				 {\it refinement for} & vertex-only & Subsection \ref{sub:vertex-cograph} & equivalent to \cite{CPS85}\\
			\cline{2-4}
			cographs & edge-only & Subsection \ref{sub:edge-cograph} & equivalent to \cite{SS04}  \\
		\hline
				 {\it refinement for} & vertex-only & Subsection \ref{sub:vertex-3leaf} & new \\
			\cline{2-4}
			$3$-leaf powers  & edge-only & Subsection \ref{sub:edge-3leaf}& new  \\
			 \hline
		\end{tabular}
\end{table}

%------------------------------------------------------------------------------------------------------------------
%------------------------------------------------------------------------------------------------------------------
\section{Graph-labelled trees, split and modular decompositions}

The purpose of this section is to introduce the notion of \emph{graph-labelled tree} and to show that the theory of split decomposition~\cite{Cun82} as well as the theory of modular decomposition~\cite{Gal67} can be stated within this framework. Before that, let us first introduce the basic terminology.

In the paper, every  graph $G=(V(G),E(G))$, or $G=(V,E)$ when clear from context,  is simple and loopless. For a subset $S\subseteq V(G)$, $G[S]$ is the subgraph of $G$ induced by $S$. If $T$ is a tree and $S$ a subset of leaves of $T$, then $T(S)$ is the smallest subtree of $T$ spanning the leaves of $S$. If $x$ is a vertex of $G$ then $G-x=G[V(G)-\{x\}]$. Similarly if $x\notin V(G)$, $G+(x,S)$ is the graph $G$ augmented by the new vertex $x$ adjacent to $S\subseteq V(G)$. Similarly if $x$ and $y$ are two vertices of $G$ such that $xy\not\in E(G)$ (resp.  $xy\in E(G)$), then define $G+e=G'(V(G),E(G)\cup\{e\})$ (resp. $G-e=G'(V(G),E(G)\setminus\{e\})$) with $e=xy$. We denote $N(x)$ the neighborhood of a vertex $x$. The neighborhood of a set $S\subseteq V(G)$ is $N(S)=\{x\notin S\mid \exists y\in S, ~xy\in E(G)\}$. The \emph{clique} is the complete graph and the \emph{star} is the complete bipartite graph $K_{1,n}$. The universal vertex of the star is called its \emph{centre} and the degree one vertices its \emph{degree-1 vertices}.
Edges of a tree will be called \emph{tree-edges}, and internal vertices of a tree $T$ will be called \emph{nodes}.

%------------------------------------------------------------------------------------------------------------------
\subsection{Graph-labelled trees}

\begin{definition}
A \emph{graph-labelled tree} $(T,\mathcal{F})$ is a tree $T$ in which every node $v$ of degree $k$ is labelled by a graph $G_v\in\mathcal{F}$ on $k$ vertices, called \emph{marker-vertices}, such that there is a bijection $\rho_v$ from the tree-edges of $T$ incident to $v$ to the marker-vertices of $G_v$. 
If $\rho_v(e)=q$ then $q$ is called an \emph{extremity} of $e$.
\end{definition}

Let $(T,\mathcal{F})$ be a graph-labelled tree and $l$ be a leaf of $T$.
A node or a leaf $u$ different from $l$ is \emph{$l$-accessible} if for every tree-edges $e=wv$ and $e'=vw'$ on the $l,u$-path in $T$, we have $\rho_v(e)\rho_v(e')\in E(G_v)$. By convention, the unique neighbor of the leaf $l$ in $T$ is also $l$-accessible.
See Figure~\ref{fig:ex_split_tree} for an example.%

\begin{definition} \label{def:accessibility}
The \emph{accessibility graph}  of  a \emph{graph-labelled tree} $(T,\mathcal{F})$ is the graph $\G(T,\mathcal{F})$ whose vertex set is the leaf set  of $T$, and in which there is an edge between $x$ and $y$ if and only if $y$ is $x$-accessible. In this setting, we say that
$(T,\mathcal{F})$ \emph{is a graph-labelled tree} of $\G(T,\mathcal{F})$.
\end{definition}

An example of a graph-labelled tree and its accessibility graph is given on Figure~\ref{fig:ex_split_tree}. 
We often abuse the language and call a leaf of $T$ a vertex of the accessibility graph and vice versa if convenient. 

\begin{figure}[htbh] 
\inserfig{0.95}{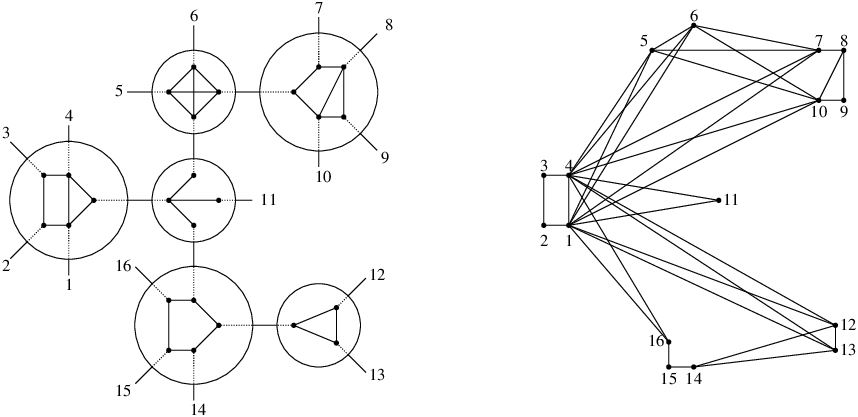}
%\centerline{\includegraphics[width=0.95\linewidth]{ex_split_tree.eps}}
\caption{A graph-labelled tree and its accessibility graph. The leaf  $12$ is $4$-accessible (and vice-versa), hence vertices $4$ and $12$ are adjacent in the accessibility graph. Every node is $4$-accessible. 
}
\label{fig:ex_split_tree}
\end{figure}

\begin{lemma} \label{lem:connected}
Let  $(T,\mathcal{F})$ be a graph-labelled tree. The accessibility graph $\G(T,\mathcal{F})$ is connected if and only if for every node $v$ of $T$ the graph $G_v\in \mathcal{F}$ is connected.
\end{lemma}

\begin{proof}
Assume there is a node $v$ of $T$ such that $G_v$ is not connected and let $C_v$ be a connected component of $G_v$. Let $L$ be the set of leaves belonging to a subtree attached to a marker-vertex of $C_v$. Then by Definition~\ref{def:accessibility}, for any leaf $l'\notin L$, none of the leaves of $L$ is $l'$-accessible.
Thereby in $\G(T,\mathcal{F})$, the set of vertices in $L$ is disconnected from the rest of the graph.

Assume for every node $v$, the graph-label $G_v$ is connected. We prove that $G=\G(T,\mathcal{F})$ is connected by induction of the number $k$ of nodes of $T$. If $k=1$, this is obviously true since $\G(T,\mathcal{F})$ and $G_v$ are isomorphic, where $v$ is the only node of $T$. Assume that $T$ contains $k>1$ nodes. Let $u$ be a node such that all its neighbors but one, say $v$, are leaves (there always exists such a node). 
Let $p$ be the marker-vertex of $G_v$ such that $\rho_v(uv)=p$. Let $(T',\mathcal{F}')$ be the graph-labelled tree obtained from $(T,\mathcal{F})$ by replacing $u$ and its leaves by a new leaf $l_u$. Notice that by construction, every leaf $l$ such that $p$ is $l$-accessible is  $l_u$-accessible.  Observe that $G$ is obtained from $G'=\G(T',\mathcal{F}')$ as follows: $V(G)=V(G')\setminus\{l_u\}\cup L_u$, where $L_u$ is the set of leaves attached to $u$ in $T$; every vertex $x\in L_u$ such that $p$ was $x$-accessible in  $(T,\mathcal{F})$ is adjacent to every neighbor of $l_u$ in $G'$; the adjacencies between the new vertices are those defined by $G_u$. As by assumption both $G'$ (induction hypothesis) and $G_u$ are connected, $G$ is also connected.
\end{proof}

\medskip
From now on, unless explicitly stated, we consider connected graphs (i.e. the graphs belonging to $\mathcal{F}$ in a graph-labelled tree $(T,\mathcal{F})$ are also connected, by Lemma \ref{lem:connected}). The next lemma is central to proofs of further theorems.

\begin{lemma} \label{lem:accessible}
Let $(T,\mathcal{F})$ be a graph-labelled tree of a connected graph $G$ and let $v$ be a node of $T$. Then every maximal tree of $T-v$ contains a leaf $l$ such that $v$ is $l$-accessible.
\end{lemma}

\begin{proof}
Let $u$ be a neighbor of node $v$ in $T$ and $T_u$ be the maximal tree of $T-v$ containing $u$. The property trivially holds if $u$ is a leaf. So assume $T_u$ contains $k\geqslant 1$ (non-leaf) nodes. If $u$ is the only node of $T_u$, as $G_u$ is connected, there exists a leaf $l$ neighboring $u$ such that the marker-vertex $\rho_u(lu)$ is adjacent in $G_u$ to the marker-vertex $\rho_u(uv)$. Thereby $v$ is $l$-accessible. Assume by induction that the property is satisfied for every tree with $k'<k$ nodes.  As $G_u$ is connected, $u$ has a neighbor $w\neq v$ such that $\rho_u(uv)$ and $\rho_u(uw)$ are  adjacent in $G_u$. Let $T_w$ be the maximal tree of $T_u-u$ containing $w$. By induction hypothesis, $T_w$ contains a leaf $l$ to which $u$ is $l$-accessible. By the  choice of $w$, $v$ is also $l$-accessible.
\end{proof}

\begin{corollary} \label{cor:neighbor-node}
Let $(T,\mathcal{F})$ be a graph-labelled tree of a connected graph $G$. Let $l$ be a leaf of $T$, and $e=uv$, $e'=uv'$ be distinct tree-edges such that $u$ is a $l$-accessible and $e$ belongs to the $u,l$ path in $T$. Then $\rho_u(e)\rho_u(e')\in E(G_u)$ if and only if there exists a $l$-accessible leaf $l'$ in the maximal tree $T_{v'}$ of $T-e'$ containing~$v'$.
\end{corollary}
\begin{proof}
If there exists a $l$-accessible leaf $l'$ in the maximal tree of $T-e'$ containing~$v'$, then by Definition~\ref{def:accessibility}, we have $\rho_u(e)\rho_u(e')\in E(G_u)$. So assume $\rho_u(e)\rho_u(e')\in E(G_u)$. 
By Lemma~\ref{lem:accessible}, $T_{v'}$ contains a leaf $l'$ such that $u$ is $l'$-accessible. As $u$ is also $l$-accessible, then $l'$ is $l$-accessible.
\end{proof}

\medskip
The above Corollary \ref{cor:neighbor-node} can be rephrased as follows: if $u$ and $v$ are two adjacent $l$-accessible nodes, then there exists a $l$-accessible leaf $l'$ such that the $l,l'$-path contains the tree-edge $uv$.

\begin{corollary} \label{cor:subgraph}
Let $(T,\mathcal{F})$ be a graph-labelled tree of a connected graph $G$. Then every graph $G_v\in\mathcal{F}$ is isomorphic to an induced subgraph of $G$.
\end{corollary}

\begin{proof}
Let $u_1,\dots u_k$ be the neighbors of node $v$ in $T$ and $T_1,\dots T_k$ be the corresponding maximal trees of $T-v$. By Lemma~\ref{lem:accessible}, for all $i$, $1\leqslant i\leqslant k$, the subtree $T_i$ of $T$ contains a leaf $l_i$ such that $v$ is $l_i$-accessible. It follows that the induced subgraph $G[\{l_1\dots l_k\}]$ is isomorphic to $G_v$.
\end{proof}

\medskip
Let $(T,\mathcal{F})$ be a graph-labelled tree of a graph $G$. Let us observe that a graph-labelled tree of any induced subgraph $H=G[X]$ can be retrieved from $(T,\mathcal{F})$. Let $T(X)$ be the smallest subtree of $T$ with set of leaves $X$. For any $G_v\in \mathcal{F}$ labelling a node $v$ of $T'$, let $G'_v$ be the subgraph induced by the marker-vertices associated with tree-edges belonging to $T'$. Then set $\mathcal{F}_X=\{G'_v\mid v\in T(X)\}$ and for every $v\in T(X)$, $\rho'_v$ is the bijection between the tree-edges of $T(X)$ incident to $v$ and the vertices of $G'_v$ such that $\rho'_v(e)=p$ if and only if  $\rho_v(e)=p$. By construction we have $\G(T(X),\mathcal{F}_X)=H$.
%EME-19-12-10 
Notice that the degree two nodes of $T(X)$ can be removed by contracting one of their two incident tree-edges.

%------------------------------------------------------------------------------------------------------------------
\subsection{Split decomposition}

\begin{definition}{\rm \cite{Cun82} \label{def:split}}
A \emph{split} of a graph $G$ is a bipartition $(V_1,V_2)$ of $V(G)$ such that 1) $|V_1|\geqslant 2$ and $|V_2|\geqslant 2$; and 2) every vertex of $N(V_1)$ is adjacent to every vertex of $N(V_2)$.
\end{definition}

A graph is \emph{degenerate} (with respect to the split decomposition) if every partition of its set of vertices into two non-singleton parts is a split.
The only degenerate graphs are known to be the cliques and the stars. 
A graph without any split is called \emph{prime} (with respect to the split decomposition).

The split decomposition of a graph $G$, as originally studied in \cite{Cun82}, consists of: finding a split $(V_1,V_2)$, decomposing $G$ into $G_1=G[V_1\cup\{x_1\}]$, with $x_1\in N(V_1)$ and $G_2=G[V_2\cup\{x_2\}]$ with $x_2\in N(V_2)$, 
$x_1$ and $x_2$ being called \emph{split-marker-vertices}; 
and then recursevily decomposing $G_1$ and $G_2$. When the process stops, the resulting graphs are called \emph{components} of the split decomposition. Adding, at each decomposition step, an edge between the pair of split-marker-vertices yields \emph{split decomposition graph}. Though the idea of a tree decomposition appears in~\cite{Cun82}, Cunningham's main result states the uniqueness of the set of  components of a split decomposition but does not focus on the structure linking them together. 
As we will see, the graph-labelled tree framework yields a natural formulation of Cunningham's result in terms of tree. 
To clarify the link between the two representations, let us point out that the split-marker-vertices in the above terminology will correspond in our setting in terms of graph-labelled trees to the marker-vertices which are extremities of internal tree-edges.

\begin{lemma} \label{lem:split-edge}
Let $(T,\mathcal{F})$ be a graph-labelled tree with no binary node and $T_1$, $T_2$ be the maximal trees of $T-e$ where $e$ is a tree-edge non-incident to a leaf. Then the bipartition $(L_1,L_2)$ of the leaves of $T$, with $L_i$ being the leaf set of $T_i$ for $i\in\{1,2\}$,
and assuming $|L_i|>1$, defines a split in the graph $\G(T,\mathcal{F})$.
\end{lemma}

\begin{proof}
Let $e=t_1t_2$ and let $l_1$ and $l_2$ be leaves of $L_1$ and $L_2$ respectively. By definition of $\G(T,\mathcal{F})$, $l_1$ and $l_2$ are adjacent if and only if $t_2$ is $l_1$-accessible and $t_1$ is $l_2$-accessible. It follows that $(L_1,L_2)$ defines a split of $\G(T,\mathcal{F})$.
\end{proof}

\medskip

\begin{figure}[h]
\inserfig{0.7}{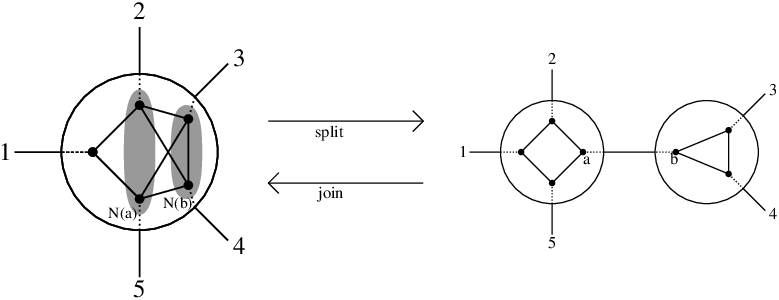}
%\centerline{\includegraphics[width=0.7\linewidth]{figures/split_join.eps}}
\caption{The node-split and the node-join operations on a graph-labelled tree.}
\label{fig:split-join}
\end{figure}

We can naturally define the \emph{node-split} operation and its converse, the \emph{node-join}, on a graph-labelled tree $(T,\mathcal{F})$ as follows (see Figure~\ref{fig:split-join}):
\begin{itemize}
\item[$\bullet$] 
\emph{Node-split in $(T,\mathcal{F})$:} Let $v$ be a node of $T$ whose graph $G_v$ has a split $(A,B)$.
Let $G_{A}$ and $G_{B}$ be the subgraphs resulting from the split $(A,B)$ of $G_v$ and $a$, $b$ be the respective split-marker-vertices. Splitting the node $v$ consists of substituting $v$ by two adjacent nodes $v_A$ and $v_B$, respectively labelled by $G_{A}$ and $G_{B}$, such that for every $p\in V(G_A)$ different from $a$, $\rho^{-1}_{v_A}(p)=\rho^{-1}_{v}(p)$ and $\rho^{-1}_{v_A}(a)=v_Av_B$ (similarly for every $q\in V(G_B)$ different from $b$, $\rho^{-1}_{v_B}(q)=\rho^{-1}_{v}(q)$ and $\rho^{-1}_{v_B}(b)=v_Av_B$).

%EME-19-12-10 j'ai ajoute contracting, mais pas e=uv ci-dessous

\item[$\bullet$] 
\emph{Node-join in $(T,\mathcal{F})$:} Let $uv$ be a tree-edge of $T$. Then joining the nodes $u$ and $v$ consists of contracting the tree-edge $uv$ and substituting $u$ and $v$ by a single node $w$ labelled by the graph $G_w$ defined as follows:
$$V(G_w)=(\ V(G_u)-\{\rho_u(uv)\}\ )\ \cup\ (\ V(G_v)-\{\rho_v(uv)\}\ )$$
$$E(G_w)= \Bigl(\ \bigl(\ E(G_u)\cup E(G_v)\ \bigr)\ \cap\ \bigl(\ V(G_w)\times V(G_w)\ \bigr)\ \Bigr)\  \cup\  \Bigl(\ N_{G_v}\bigl(\rho_v(uv) \bigr)\ \times\ N_{G_u}\bigl(\rho_u(uv)\bigr)\ \Bigr)$$
For every marker-vertex $p\in V(G_w)$, $\rho^{-1}_w(p)=\rho^{-1}_v(p)$ if $p\in V(G_v)$ and  $\rho^{-1}_w(p)=\rho^{-1}_u(p)$ if $p\in V(G_u)$.
\end{itemize}
\medskip

Observe that if $(T,\mathcal{F})$ is obtained from $(T',\mathcal{F}')$ by a node-join or a node-split operation, then it follows from the definitions that $\G(T,\mathcal{F})=\G(T',\mathcal{F}')$. This show that a given graph is not representated by a unique graph-labelled tree.

\begin{figure}[h]
\inserfig{0.55}{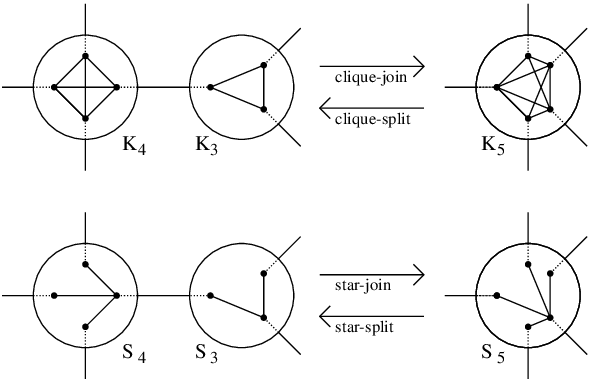}
%\centerline{\includegraphics[width=0.55\linewidth]{figures/clique_and_star_join.eps}}
\caption{Node-split and node-join operations on cliques and stars.}
\label{fig:clique-star-join}
\end{figure}

Among the node-join operations, let us distinguish: the \emph{clique-join}, operating on two neighboring nodes labelled by cliques, and the \emph{star-join}, operating on star-labelled neighboring nodes $u$, $v$ such that the tree-edge $uv$ links the centre of one star to a degree-1 vertex of the other. The converse operations are called respectively 
\emph{clique-split} and \emph{star-split}.  See Figure \ref{fig:clique-star-join}.
Also, if a node $v$ of a graph-labelled tree has degree $2$ in a graph-labelled tree, then $G_v$  consists of an edge between two marker vertices and thereby $v$ can be contracted without loss of information.
A graph-labelled tree 
$(T,\mathcal{F})$ is \emph{reduced} if every node has degree $>2$ and neither a clique-join nor a star-join can be applied. So hereafter we only consider graphs with at least $3$ vertices.

We are now able to reformulate the main split decomposition theorem first established in~\cite{Cun82}.
For completeness of the paper, a direct proof of Theorem \ref{th:cun} in terms of graph-labelled trees is provided in the appendix.

\begin{theorem}[Cunningham's Theorem reformulated]  \label{th:cun}
For every connected graph $G$, there exists a unique reduced graph-labelled tree $(T,\mathcal{F})$ such that  $G=\G(T,\mathcal{F})$ and every graph of $\mathcal{F}$ is prime or degenerate.
\end{theorem}

For a connected graph $G$, the \emph{split tree} $ST(G)$ of $G$ is the unique reduced graph-labelled tree $(T,\mathcal{F})$ in the above Theorem \ref{th:cun}. As an example, see Figure~\ref{fig:ex_split_tree} where the graph-labelled tree is effectively reduced.

\begin{corollary} \label{cor:split_list}
Let $ST(G)=(T,\mathcal{F})$ be the split tree of a connected graph $G=(V,E)$.
Then every split of the graph $G$ is the bipartition of the set of leaves of $T$ induced by removing a tree-edge
of $(T',\mathcal{F}')$, a graph-labelled tree which is obtained from $(T,\mathcal{F})$ by at most one node-split operation on a degenerate node.
\end{corollary}

The next Lemma will be crucial for algorithm complexity means.

\begin{lemma} \label{lem:tree-size}
Let $ST(G)=(T,\mathcal{F})$ be the split tree of a connected graph $G$. For every vertex $x\in V(G)$, $T(N(x))$ has at most $2.|N(x)|$ nodes.
\end{lemma}

\begin{proof}
Let $u$ and $v$ be two adjacent nodes in $T(N(x))$ such that  $v$ has degree $2$ in $T(N(x))$ and $u$ is on the $x,v$-path. Let $a$ be the marker-vertex of $G_v$ such that $\rho^{-1}_v(a)=uv$. Then $a$ has degree $1$ in $G_v$ otherwise, by Corollary \ref{cor:neighbor-node}, node $v$ would have degree $>2$. Hence $G_v$ is not prime (a graph with a pendant vertex has a split), hence it is a star with centre $b$ such that $ab$ is an edge of $G_v$. Let $w$ be the node neighbor of $v$ such that $\rho^{-1}_v(b)=vw$. If $w$ is not a leaf, then $w$ has degree $>2$ in $T(N(x))$, otherwise it would be a star $\rho_w(vw)$ being a degree one marker-vertex and the tree would not be reduced.
So $T(N(x))$ does not contains two adjacent degree two nodes.
Hence the result.
\end{proof}

%------------------------------------------------------------------------------------------------------------------
\subsection{Modular decomposition}
\label{sec:mod-dec}

The modular decomposition of a graph is a well understood decomposition process (see~\cite{MR84} for a complete survey). However the purpose of this section is to show that the graph labelled-trees are also a natural tool to represent the modular decomposition. Thereby it provides a framework common to the split and the modular decomposition. 

\begin{definition} \label{def:module}
A \emph{module} of a graph $G$ is a set $M$ of vertices such that every vertex $x$ outside $M$ is either adjacent to all the vertices of $M$ ($M\subseteq N(x)$) or to none of them ($M\cap N(x)=\emptyset$). 
\end{definition}

Singleton vertex sets and the whole vertex set are the \emph{trivial} modules of $G=(V,E)$.
A graph is \emph{degenerate} with respect to the modular decomposition, or \emph{$M$-degenerate} (to avoid confusion with the split decomposition), if every subset of its vertices is a module. The $M$-degenerate graphs are cliques or stables (the graph with an empty edge set - or independent set). Intuitively,  cliques and \emph{stables}  play the same role with respect to the modular decomposition than cliques and stars with respect to the split decomposition.
A graph is \emph{prime} with respect to the modular decomposition, or \emph{$M$-prime}, whenever all its modules are trivial. 

If $\mathcal{P}=\{M_1,\dots M_k\}$ is a partition of the vertex set of a graph $G$, the \emph{quotient graph} $G{/\mathcal{P}}$ is defined as the unique (up to isomorphism) subgraph induced by a subset $P\subset V$ such that for all $i$, $1\leqslant i\leqslant k$, $|P\cap M_i|=1$. Each vertex $x_i\in P\cap M_i$ is called the \emph{representative} of $M_i$, for $i$, $1\leqslant i\leqslant k$.

As the split decomposition, the modular decomposition of a graph $G=(V,E)$ is commonly understood as a recursive process: 1) find a partition of the vertex set $V$ into modules say $\mathcal{P}=\{M_1,\dots M_k\}$; and 2) recursively decompose the subgraphs $G[M_i]$ for all $i$, $1\leqslant i\leqslant k$. This naturally yields a rooted tree decomposition. In 1967, Gallai~\cite{Gal67} showed that every graph $G$ has a canonical \emph{modular decomposition tree}, denoted $MD(G)$, which is obtained by choosing at the each step of the recursive process the coarsest possible partition.
The leaf set of $MD(G)$ is the vertex set of $G$ and each node is labelled by the quotient graph associated with the corresponding partition. These graph labels are either clique, stable or graphs that are $M$-prime graphs.
In the usual terminology, clique labelled nodes are called \emph{series} (or $1$-nodes) and stable labelled nodes are called \emph{parallel} nodes (or $0$-nodes). The canonicity of the modular decomposition tree results from the constraint that no series node (resp. parallel node) is a child of a series node (resp. parallel node). Two vertices $x$ and $y$ are adjacent in $G$ if and only if their representative vertices are adjacent in the quotient graph $G/\mathcal{P}$.
Figure~\ref{fig:md-split} shows an example of a graph and its modular decomposition tree.

\begin{figure}[h]
\inserfig{0.95}{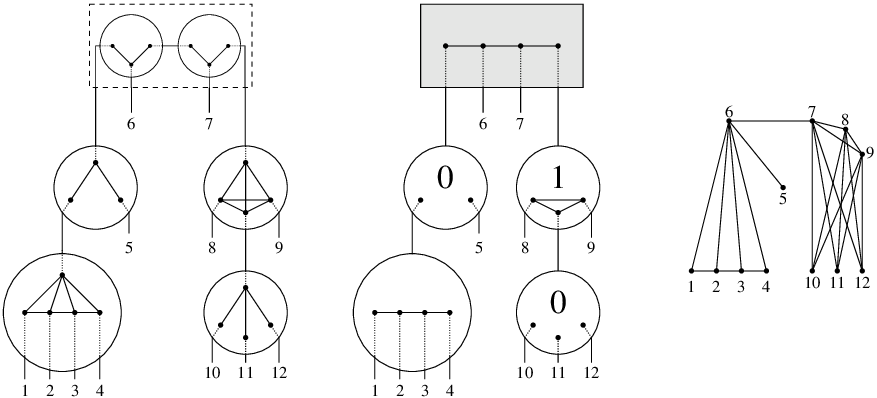}
\caption{A graph on the right, its modular decomposition tree in the middle, and its split tree on the left. 
The node in a larger circle is prime in each decomposition.
The grey squared node, call it $u$, corresponds to the root of the modular decomposition tree.
Node $u$ is $M$-prime, but not prime for the split decomposition. 
Observe that a node-split on $u$ yields stars whose centres are not towards each other.
The modular graph-labelled tree is obtained by the converse node-join operation, i.e. replacing the dashed squared subtree of the split tree by $u$.}
\label{fig:md-split}
\end{figure}

Let us now describe how the modular decomposition tree $MD(G)$ of a connected graph $G$ naturally transforms into a reduced graph labelled tree $(T_M,\mathcal{F}_M)$ whose accessibility graph is $G$ (see Figure~\ref{fig:md-split}):

\begin{enumerate}
\item Unless the root of $MD(G)$ has degree two, $T_M$ is isomorphic to the tree underlying $MD(G)$. If $MD(G)$ has a binary root, then $T_M$ is isomorphic to the tree resulting from the contraction in $MD(G)$ of one of the tree-edges incident to the root.

\item For a node $u$, distinct from the root of $MD(G)$, 
with associated quotient graph $G{/\mathcal{P}_u}$ labelling $u$ in $MD(G)$,
the label $G_u$ in $(T_M,\mathcal{F}_M)$ is obtained by adding a universal marker-vertex to
$G{/\mathcal{P}_u}$ which is mapped to the tree-edge $uv$ where $v$ is the father of $u$ in $MD(G)$.
\end{enumerate}

Note that if $u$ is a parallel node in $MD(G)$, then it becomes a star node in $(T_M,\mathcal{F}_M)$. It is straightforward to see from the definitions that $G$ is the accessibility graph of $(T_M,\mathcal{F}_M)$. 
Let us also point out that the root node of $MD(G)$ is binary if $G$ has a universal vertex $x$ and $G-x$ is $M$-prime or if $\overline{G}$ is the disjoint union of two connected components. 
Finally, $(T_M,\mathcal{F}_M)$ is reduced since two series nodes or two parallel nodes are not adjacent in the modular decomposition tree. We will call  \emph{modular graph-labelled tree} this graph-labelled tree $(T_M,\mathcal{F}_M)$.

We can now reformulate Gallai's theorem~\cite{Gal67} in term of graph-labelled trees.

\begin{theorem}[Gallai's Theorem reformulated]  \label{th:gallai}
For every connected graph $G$, there exists a unique reduced graph-labelled tree $(T_M,\mathcal{F}_M)$ with $G=\G(T_M,\mathcal{F}_M)$ such that $T_M$ contains a node or a tree-edge $r$, called the \emph{root}, and for every node $v\neq r$, we have
1) the graph $G_v$ contains a universal vertex $x$ such that $G_v-x$ is $M$-prime or $M$-degenerate, 
and 2) the tree-edge associated with $x$ in $T_M$ is on the path between $v$ and $r$. 
\end{theorem}

\begin{lemma}\label{rk:mod-to-split}
Let $G$ be a connected graph. In $MD(G)$, the label of a non-root node $u$ is $M$-prime if and only if its corresponding label in the modular graph-labelled tree $(T_M,\mathcal{F}_M)$ is prime for the split decomposition.
\end{lemma}

\begin{proof}
Follows from the definitions of split and module, and from the construction above.
\end{proof}
\bigskip

Using %It follows from %Remark \ref{rk:mod-to-split} 
Lemma \ref{rk:mod-to-split} we can describe how the split tree and the modular graph labelled tree can be retrieved from each other:

\begin{itemize}
\item \emph{From the modular graph labelled tree $(T_M,\mathcal{F}_M)$ to the split tree $ST(G)$}: If the root of $T_M$ is not a node, then $ST(G)=(T_M,\mathcal{F}_M)$. If the root of $T_M$ is a node $u$, then 
substitute the split tree of $G_u$ to node $u$ (i.e. node-split $(T_M,\mathcal{F}_M)$ according to the splits of $G_u$ and lastly make clique-joins or star-joins to get a reduced graph labelled tree).

\item \emph{From the split tree $ST(G)$ to the modular graph labelled tree $(T_M,\mathcal{F}_M)$}: If $ST(G)=(T,\mathcal{F})$ contains at least two node, then pick
a node $u$  such that every incident tree-edge but one, say $e$, is adjacent to a leaf, test if $\rho_u(e)$ is a universal vertex of $G_u$.
If so, then delete $u$ from $T$ (i.e. replace it with a leaf) and
repeat until no deletion is possible.
The set of remaining nodes induces a subtree $T'$ of $T$. Then $(T_M,\mathcal{F}_M)$ results from the series of node-joins applied on each internal tree-edge of $T'$
(i.e. substituting a single node labelled by the accessibility graph of $T'$ to $T'$).
\end{itemize}

It is worth to notice that a subtree of the split tree, namely $T'$, plays the role of the root of the modular decomposition tree, 
though, unlike the modular decomposition tree, the split tree is fundamentally unrooted.
Figure \ref{fig:md-split} illustrates these two decompositions on an example.

%------------------------------------------------------------------------------------------------------------------
%------------------------------------------------------------------------------------------------------------------
\section{Split tree characterizations of restricted graph classes}

This section presents bijective and incremental characterizations of distance hereditary graphs, cographs and 3-leaf power graphs, in terms of their split tree. The characterization of distance hereditary graphs yields an intersection model which answers an open question (see~\cite{Spi03}, page 309). Incremental characterizations of each of these three graph classes are also derived. Such a result was already known for cographs~\cite{CPS85} (based on the modular decomposition tree), but not for distance hereditary graphs neither for $3$-leaf powers. These characterizations will be the basis of the vertex-only fully-dynamic recognition algorithms developed in Section 4.

%------------------------------------------------------------------------------------------------------------------
\subsection{Distance hereditary graphs}

\begin{definition}
A graph $G$ is \emph{distance hereditary} (DH for short) if for every connected subgraph $H$ of $G$, the distance between any two vertices $x$ and $y$ in $H$ is the same than the distance between $x$ and $y$ in $G$.
\end{definition}

A graph is \emph{totally decomposable} by the split decomposition if every induced subgraph with at least $4$ vertices contains a split. By \cite{HM90}, it is known that a graph is DH if and only if  it is totally decomposable by the split decomposition, i.e. the nodes of its split tree are labelled by cliques and stars. Hence DH graphs are exactly accessibility graphs of clique-star labelled trees,  \emph{clique-star trees} for short. Among the possible clique-star trees, the split tree is the unique reduced one. In other words, there is a bijection between DH graphs and reduced clique-star trees. Figure~\ref{fig:ex_split_tree_DH} gives an example. We mention that ternary clique-star trees were used in~\cite{EGY05} to draw DH graphs.

\begin{figure}[h]
\inserfig{0.95}{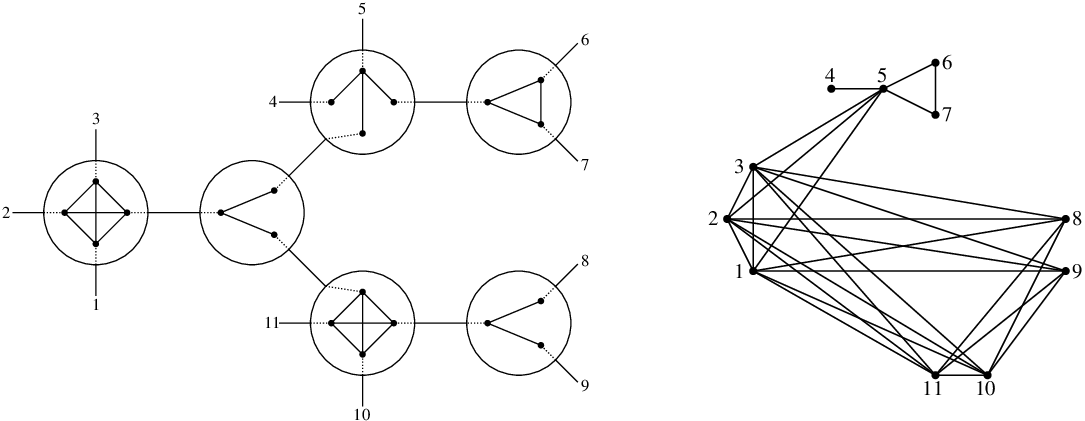}
%\centerline{\includegraphics[width=0.95\linewidth]{figures/ex_split_tree_DH.pdf}}
\caption{A clique-star reduced tree and its accessibility DH graph.}
\label{fig:ex_split_tree_DH}
\end{figure}

Let us notice that the classical construction
of DH graphs~\cite{BM86} (there exists a linear ordering for vertex-insertion such that each new vertex $y$ is (a) true twin, (b) false twin, or (c) pendant)
is easy to read on the clique-star tree, see Figure \ref{fig:incremental}. We also mention that DH graphs can be characterized by forbidden induced subgraphs~\cite{BM86} (see Section 5 for details).

\begin{figure}[h]
\inserfig{0.55}{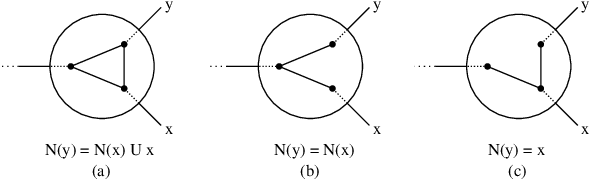}
%\centerline{\includegraphics[width=0.55\linewidth]{figures/incremental.pdf}}
\caption{Usual (static) incremental construction of DH graphs: a) adding a true twin $y$ of vertex $x$ amounts to insert a degree $3$ clique node on the tree-edge incident to leaf $x$ and attach leaf $y$ to that node; b) adding a false twin $y$ of vertex $x$ amounts to insert a degree $3$ star node on the tree-edge incident to leaf $x$ such that $x$ and $y$ are mapped to the extremities of the star; and c) adding a pendant vertex $y$ to vertex $x$ amounts to insert a degree $3$ star node whose centre is mapped to $x$ and to which $y$ is attached.}
\label{fig:incremental}
\end{figure}

In what follows, we will call simply \emph{clique node}, resp. \emph {star node}, a clique labelled node, resp. star labelled node. 

%--------------------------------------------------------------------------
\medskip
\noindent
{\textbf{\em An intersection model.}}
Given a family $S$ of sets, one can define the intersection graph $\mathcal{I}(S)$ as the graph whose vertices are the elements of $S$ and there is an edge between two elements if and only if  they intersect.
Many restricted graph families are defined or characterized as the intersection graphs (e.g. chordal graphs, interval graphs\dots~see \cite{MM99}). Graph families supporting an intersection model can be characterized without even specifying the model~\cite{MM99}. This result applies to DH graphs, but no model has been yet given (see \cite{Spi03}, page 309). Based on clique-star trees, an intersection model can be easily derived. Note that it can be equivalently stated by considering only reduced clique-star trees, or even only ternary ones. We call \emph{accessibility set} of a leaf $l$ in a graph-labelled tree the set of pairs $\{l,l'\}$ with $l'$ a $l$-accessible leaf, or, equivalently, the set of paths in the tree joining $l$ to a $l$-accessible leaf $l'$.
Notice that an accessibility set could also be defined as the set  of paths in the clique-star tree from a given leaf to its accessible leaves.

\begin{theorem}[Intersection model] \label{th:intersection-model}
A  graph is distance hereditary if and only if it is the intersection graph of a family of accessibility sets of leaves in a set of clique-star trees.
\end{theorem}

\begin{proof}
Follows directly from the representation of DH graphs as accessibility graphs of clique-star trees.
\end{proof}

\medskip
Observe that finding an intersection model always amounts to characterize adjacencies in terms of an independent structure  (in our case the clique-star trees) in which some objects correspond to vertices and any arbitrary set of those objects induces a graph belonging to the required graph class. 
In that sense, our intersection model can be compared with other well-known intersection models. For example, consider the subtrees of a tree model of chordal graphs~\cite{Gav74}. This model could be derived from the characterization of chordal graphs as the set of graphs having a tree-decomposition~\cite{RS86} in which every node induces a clique. Likewise, our DH intersection model derives from the fact that DH graphs are the graphs whose split tree is a clique-star tree. Both models rely on some tree-like structure. In the model of chordal graphs, the subtrees represent the 
interlacing
structure of the sets $\mathcal{C}_x$ of clique bags, where, for each vertex $x$, $\mathcal{C}_x$ is the set of bags containing $x$. 
In the DH model the accessibility sets represent the 
interlacing structure of the sets of alternating paths with a common leaf in the tree,
depending on the way cliques and stars are spread over the nodes of the tree.

%--------------------------------------------------------------------------
\bigskip
\noindent
{\textbf{\em Incremental characterization.}}
Let $G$ be a connected DH graph and let $ST(G)=(T,\mathcal{F})$ be its split tree. 
Given a subset $S$ of $V(G)$ and $x\not\in V(G)$, we want to know whether the graph $G+(x,S)$ is DH or~not. We first discard the obvious case where $|S|=1$ which consists in adding a pendant vertex $x$ attached to $y\in V(G)$. In that case, it is well known that $G+(x,S)$ is a DH graph if and only if  $G$ is.

\begin{definition} \label{def:accessible}
For $S\subseteq V(G)$, let  $T(S)$ be the smallest subtree of $T$  
with set of leaves $S$.
Let $u$ be a node of $T(S)$.
\vspace{-0.2cm}
\begin{enumerate}
\item $u$ is \emph{fully accessible} (w.r.t. $S$) if every maximal tree of $T-u$ contains a leaf $l\in S$;
\item $u$ is \emph{singly accessible} (w.r.t. $S$)  if it is a star node and exactly two maximal trees of $T-u$ contain a leaf 
$l\in S$ 
among which the maximal tree containing the neighbor $v$ of $u$ such that $\rho_u(uv)$ is the centre~of~$G_u$;
\item $u$ is \emph{partially accessible} (w.r.t. $S$) otherwise.
\end{enumerate}
\end{definition}

We say that a star node $v$ is \emph{oriented towards} a tree-edge (or a node) $t$ of $T$ if the tree-edge $e$ such that $\rho_v(e)$ is the centre of $G_v$ is on the path in $T$ between $t$ and $v$.
Figure \ref{fig:ex_caract} illustrates Definition \ref{def:accessible} above and Theorem \ref{th:charac}  below.

\begin{figure}[htbh]
	\inserfig{0.7}{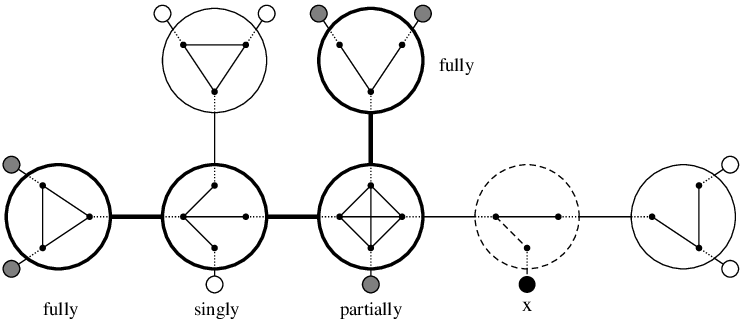}
\caption{
Consider that the dashed node is omitted, that is precisely: the dashed node and its three incident tree-edges are deleted, and replaced with a tree-edge between its two adjacent nodes. Then the figure represents the split tree $ST(G)$ of a DH graph.   Elements of $S\subseteq V(G)$ are represented as grey leaves. The subtree $T(S)$ is represented with bold nodes and tree-edges. Check that the properties of Theorem \ref{th:charac} are satisfied:
the fully accessible star node is oriented towards the unique partially accessible node,
whereas the singly accessible star node is not.
So, if a vertex $x$ is added to $G$ with neighborhood $S$, then the graph $G+x$ is DH. Its split tree $ST(G+x)$ is obtained by inserting the dashed node.%
}
	\label{fig:ex_caract}
\end{figure}

\begin{theorem}[Vertex incremental characterization] \label{th:charac}
Let $G$ be a connected distance hereditary graph and $ST(G)=(T,\mathcal{F})$ be its split tree. Then $G+(x,S)$, with $|S|>1$ is distance hereditary  if and only~if:
\begin{enumerate}
\item at most one node of $T(S)$ is partially accessible; 
\item every clique node of $T(S)$ is either fully or partially accessible;
\item if there exists a partially accessible node $u$ in $T(S)$, then every star node $v\not=u$ of $T(S)$ is oriented towards $u$ if and only if it is fully accessible;
otherwise, there exists a tree-edge $e$ of $T(S)$ towards which every star node of $T(S)$ is oriented if and only if it is fully accessible.
\end{enumerate}
\end{theorem}

\begin{proof}
\begin{itemize}
\item[$\Rightarrow$] Since $G+(x,S)$ is a DH graph, it is the accessibility graph of  a ternary clique-star tree $(\tilde{T},\tilde{\mathcal{F}})$. Let $u$ be the node of $\tilde{T}$ to which $x$ is attached and let $v$, $w$ be its neighbors. Now consider the clique-star tree  $(T',\mathcal{F}')$ obtained by applying every possible clique-join or a star-join to tree-edges different from $uv$ and $uw$. Notice that $ST(G)$ is obtained by 1) removing the leaf $x$ and the marker vertex $\rho_u(xu)$, 2) performing a node-join to get rid of the degree two node $u$ thereby creating a tree-edge $vw$, and 3) if needed apply a node-join on the tree-edge $vw$.

Assume the node-join on $vw$ is not required to obtain $ST(G)$. Then every node of $T(S)$ is a node of $T'$. By construction, every leaf of $S$ is $x$-accessible in $(T',\mathcal{F}')$. Then the three conditions are a consequence of Corollary~\ref{cor:neighbor-node}. Precisely, observe that if $u$ is a clique node, then $T(S)$ does not contain any partially accessible node, every star node is oriented towards the tree-edge $vw$ if and only if it is fully accessible. If $u$ is a star node, then $\rho_u(xu)$ is a degree-1 marker vertex. In that case, if $\rho_u(uv)$ is the centre $G_u$ and $v$ is a star node, then $v$ is the only partially accessible node in $T(S)$ (the case $\rho_u(uw)$ is the centre $G_u$ and $w$ is a star node is symmetric). 

Assume $ST(G)$ is obtained after a node-join on $vw$ which results on a new node $u'$. Then every node of $T(S)$ except $u'$  corresponds to a node of $T'$. Again by Corollary~\ref{cor:neighbor-node} the nodes of $T(S)$ different than $u'$ are all singly or fully accessible, and a star node is oriented towards $u'$ if and only if it is fully accessible. 
If $x$ is adjacent to a star node $u$ in $T'$, or if $x$ is adjacent to a clique node $u$ in $T'$ and $u'$ is a star, then it is straightforward to check that $u'$ is
partially accessible and the conditions are satisfied. 
If $x$ is adjacent to a clique node $u$ in $T'$ and $u'$ is a clique, then $u'$ is
fully accessible and a star node is oriented towards any tree-edge incident to $u'$ if and only if it is fully accessible, so the conditions are satisfied. 

\item[$\Leftarrow$] Assume there is no partially accessible node. So there exists a tree-edge $e=uv$ of $T(S)$ towards which the star nodes of $T(S)$ are oriented if and only if they are fully accessible. 
Let $(T',\mathcal{F}')$ be the clique-star tree obtained by: 1) subdividing $e=uv$ into $e_u=uw$ and $e_v=wv$; 2) attaching the leaf $x$ to $w$ (which is thereby a ternary node); 3) making $w$ a clique node if the two maximal trees of $T-e$ contain a leaf of $S$, otherwise $w$ is a star node whose centre is $\rho_w(wu)$.

Every node of $T(S)$ is either fully accessible or singly accessible, a node of degree $2$ in $T(S)$ is singly accessible. Let $w'$ be a node on the path in $T$ between any $y\in S$ and $e$ and let $e_y$, $e_x$ be the two tree-edges of that path incident to $w'$. By Definition~\ref{def:accessible}, we have that $\rho_{w'}(e_x)\rho_{w'}(e_y)\in E(G_{w'})$. It follows that every $y\in S$ is a neighbor of $x$ in $\G(T',\mathcal{F}')$. Let us now prove that every $z\notin S$ is not a neighbor of $x$ in $\G(T',\mathcal{F}')$, thereby proving that $\G(T',\mathcal{F}')=G+(x,S)$. Let $w'$ be the node of $T(S)$ which is the closest to the leaf $z$, and let $e_{w'}$ be the tree-edge incident to $w'$ in the path between $w'$ and $z$. 
By the choice of $w'$, $w'$ cannot be fully accessible (otherwise it would not be the closest to $z$). So $w'$ is singly accessible and thereby is a star node. Its centre is not oriented towards $e$ by condition 3, and not oriented towards $e_{w'}$ by Definition \ref{def:accessible}. It follows that the neighbor $w''$ of $w'$ on the path between $w'$ and $e$ is not $z$-accessible. Thus $z$ is not a neighbor of $x$ in $\G(T',\mathcal{F}')=G+(x,S)$.

Assume there is a partially accessible node $u$. Then it suffices to node-split the node $u$ into two new nodes $v$ and $w$, such that $v$ is adjacent to the neighbors of $u$ not belonging to $T(S)$ and $w$ to those belonging to $T(S)$. Now star-nodes of $T(S)$ are oriented towards the new tree-edge $e=vw$, and the same construction and arguments than above apply.
\end{itemize}%
Note that the complete and detailed case by case description of the constructions involved in this proof is made in the algorithmic Section \ref{sec:vertex-algo}.%
\end{proof}

%----------------------------------------------------------------------------------------------------------------------------
\subsection{A split decomposition characterization of cographs}

A cograph is a $P_4$-free graph~\cite{Sum73} (see Figure~\ref{fig:p4}). This graph family is also known as the graphs totally decomposable by the modular decomposition: i.e. their modular decomposition tree does not contain any $M$-prime node. Moreover cographs are known to be DH graphs.

\begin{figure}[htbh]
\inserfig{0.5}{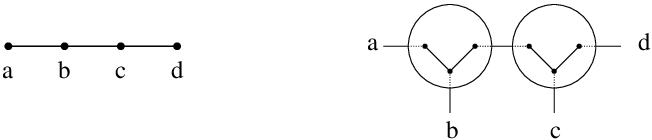}
%\centerline{\includegraphics[width=0.5\linewidth]{figures/p4.pdf}}
\caption{The $P_4$ is the smallest graph that is not a cograph. Although its split tree only contains star nodes, there is no tree-root toward which the stars are oriented.}
\label{fig:p4}
\end{figure}

\begin{theorem}[Cograph split tree characterization]\label{th:charact-cograph}
A connected graph $G=(V,E)$ is a cograph if and only if its split tree $ST(G)$ is its modular graph-labelled tree and is a clique-star tree.
\end{theorem}
\begin{proof}
Assume that $G$ is a cograph. By Theorem~\ref{th:gallai}), $MD(G)$ does not contains any $M$-prime node, the modular graph-labelled tree of $G$ only contains clique and star nodes.
Moreover by definition $(T_M,\mathcal{F_M})$ is reduced, it is also the split tree $ST(G)$.

Assume that $G$ is not a cograph. Then the modular graph-labelled tree contains a node $u$ such that $G_u$ is neither a star nor a clique. If $G_u$ is prime with respect to the split decomposition, we are done (since then $ST(G)$ is not a clique-star tree). So assume the graph $G_u$ contains a split, then the node set of $ST(G)$ and of the modular graph-labelled tree are not the same. That ends the proof.
\end{proof}

\medskip
Thanks to the construction of the modular graph-labelled tree (see Section~\ref{sec:mod-dec}), we can rephrase Theorem~\ref{th:charact-cograph} as follows:

\begin{corollary}\label{cor:charact-cograph}
A connected graph $G=(V,E)$ is a cograph if and only if $ST(G)$ is a clique-star tree and either contains a clique node or a tree-edge towards which all the star nodes are oriented. Such a clique-node or tree-edge will be called hereafter the \emph{tree-root} of $ST(G)$.
\end{corollary}

\begin{figure}[htbh]
\inserfig{0.85}{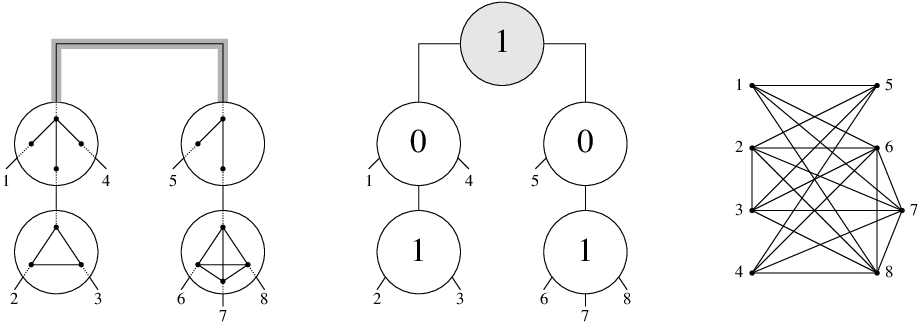}
%\centerline{\includegraphics[width=0.7\linewidth]{figures/cotree.pdf}}
\caption{A cograph on the right,  its split tree  on the  left, and its cotree in the middle. 
The star nodes, corresponding to $0$ labels in the cotree, are oriented towards the tree-root (grey tree-edge).}%
\label{fig:cotree}
\end{figure}

For the sake of simplicity, let us denote the tree-root of the split tree $ST(G)$ of a cograph by the set $R$ of nodes of $T$ it contains: that is we set $R=\{u\}$ if the $R$ is a clique-node $u$ and $R=\{u,v\}$ if the $R$ is a tree-edge $uv$ with $u$ and $v$ being star nodes.
\smallskip

%Observe that, to get a cograph vertex incremental characterization, we could simply say: \emph{$G+x$ is a DH graph} (using Theorem~\ref{th:charac}) \emph{and $ST(G+x)$ is the modular graph-labelled tree of $G+x$} (see Theorem~\ref{cor:charact-cograph}). The following theorem establishes a more precise property on $T(S)$.

Observe that, to get a cograph vertex incremental characterization, we could simply test, given a cograph $G$, first if the graph $G+x$ is a DH graph using Theorem~\ref{th:charac}, and then if the node to which $x$ is attached in $ST(G+x)$ does not create a contradiction with Corollary~\ref{cor:charact-cograph}. This second condition amounts to test a local condition on $ST(G+x)$, and would be enough
for algorithmic purpose to refine the main DH algorithm of Section  \ref{sec:vertex-algo} in terms of cographs as done in Section \ref{sub:vertex-cograph}.
However, the following theorem establishes a more precise property directly on $ST(G)$.

\begin{theorem}[Cograph vertex incremental characterization] \label{th:inc-cograph}
Let $G$ be a connected cograph and $ST(G)=(T,\mathcal{F})$ be its split tree with tree-root $R$. Then $G+(x,S)$ is a cograph if and only if:
\begin{enumerate}
\item it is a distance hereditary graph (see conditions of Theorem~\ref{th:charac}) and
\item the subtree $T(S)$ 
of $T$ either intersects $R$ or contains a node adjacent to a node of $R$.
\end{enumerate}
\end{theorem} 

\begin{proof}
As  every star-node of the split tree of a cograph is oriented towards the root, $ST(G)$ and $T(S)$ have a natural orientation. This implies that condition 2 above can be rephrased as follows: \emph{if $T(S)$ does not intersect $R$, then $T(S)$ has a unique root node which is adjacent to a node of $R$}.
\begin{itemize}
\item[$\Rightarrow$] If $G+(x,S)$ is a cograph, then it is a DH graph. By the structure of their split tree (see Theorem~\ref{th:charact-cograph}), observe that every node of the tree-root is $l$-accessible for every leaf $l$. Let us consider the three different ways $ST(G)$ can be transformed into $ST(G+(x,S))$:
\begin{enumerate}
\item \textit{Vertex $x$ has been attached to a node $u$ of $ST(G)$.} Then the tree-root $R$ of $ST(G)$ is still the tree-root of $ST(G+(x,s))$. By Corollary~\ref{cor:charact-cograph}, $R$ either contains a clique node or two star nodes $v$ and $w$ oriented towards the tree-edge $vw$ ($u$ may belong to $R$). Observe that in both cases, the nodes of $R$ are $x$-accessible. By Corollary~\ref{cor:neighbor-node}, the set $S$ intersects the leaf set of at least two maximal trees of $T-R$. Thereby $R$ intersects the node set of $T(S)$.

\item \textit{A node $u$ of $ST(G)$ is node-split into two adjacent nodes $v$ and $v'$ and the tree-edge $vv'$ is subdivided to insert a new node $w$ adjacent to $x$.} If $u$ does not belong to the tree-root $R$ of $ST(G)$, then as in the first case the tree-root remains unchanged and $R$ intersects the node set of $T(S)$. Assume $R=\{u\}$. By Corollary~\ref{cor:charact-cograph}, $u$ is a clique-node and the new node $w$ is a star node, say oriented towards $v$. Observe that every maximal subtrees of $T-u$ is now attached to either $v$ or $v'$ which both have degree at least $3$, and that by Corollary~\ref{cor:neighbor-node} each of these subtree attached to $v$ contains a leaf in $S$. Thereby $R$ belongs to the node set of $T(S)$. So assume that $R=\{u,v\}$, which implies that $u$ is a star-node (Corollary~\ref{cor:charact-cograph}). Again by Corollary~\ref{cor:neighbor-node}, $T-u$ contains at least two maximal trees of $T-u$ with a leaf in $S$ and at least one of these maximal trees is the one containing node $v$. It follows that $R$ is a subset of the node set of $T(S)$. 

\item \textit{A tree-edge of $ST(G)$ is subdivided to insert a new node $w'$ adjacent to $x$}. As clique-nodes and star-nodes alternate everywhere in $ST(G)$ but possibly at the tree-root $R=\{u,v\}$, the subdivided tree-edge is: either a) the tree-edge $uv$ joining the vertices of the tree-root; b) or is incident to a leaf $l$; or c) incident to unique node $u$ of the tree-root and $u$ is a clique-node. Let us consider these three different cases:
\begin{enumerate}
\item Assume the subdivided tree-edge is $uv$ with $R=\{u,v\}$. If the tree-root does not contains a leaf, then by Corollary~\ref{cor:charact-cograph} node $w'$ is a clique node. It follows that the two maximal trees of $T-uv$ contain leaves of $S$, implying that $R$ is a subset of the node set of $T(S)$. The tree-root of $ST(G+(x,S))$ is now $\{w'\}$. If the tree-root  $R=\{u,v\}$ contains a leaf, say $v$, then $T(S)$ either contains $v$ or $S$ contains two leaves in different maximal trees of $T-u$, which implies that the node set of $T(S)$ intersects $R$. 
\item Assume the subdivided tree-edge is $wl$ with $l$ a leaf. Then the tree-root of $ST(G+(x,S))$ is still $R$  and the same arguments than in case 1 above apply.
\item Assume the subdivided tree-edge is $uv$ with $R=\{u\}$ ($u$ is a clique-node). The node $w'$ is a star-node oriented towards the star-node $v$. In that case at least two maximal trees of $T-v$ contains leaves in $S$ and thereby $v$ belongs to $T(S)$. So we are in the situation that $T(S)$ does not intersect $R$ but has a neighbor, namely $v$, in the tree-root. 
\end{enumerate}
\end{enumerate}

\item[$\Leftarrow$] We need to show that the second condition implies that all the star nodes are oriented towards the root of $ST(G+(x,S))$ (condition 2 of Theorem~\ref{th:charact-cograph}). This is trivially the case if no new node has been created while transforming $ST(G)$ into $ST(G+(x,S))$. This is also true if a new clique-node has been created. So assume that a new star-node $w$ has been inserted. Either the node $w$ arises from the subdivision of a clique-node $u$ or from the subdivision of a tree-edge. Consider the former case. If $\{u\}$ is not the tree-root $R$ of $ST(G)$, then the tree root of $ST(G+(x,S))$ is still $R$. As nodes of the tree-root are $x$-accessible, the result follows. Otherwise if $R=\{u\}$, then the new tree root of $ST(G+(x,S))$ is one of the two new clique-nodes resulting from the subdivision of $u$. The result trivially holds. Consider now the latter case ($w$ is inserted on a tree-edge). This tree-edge has to contain a leaf, say $l$ adjacent to $u$. If $\{u,l\}$ is not the tree-root $R$, then as before, $R$ is still the tree-root of $ST(G+(x,S))$ and thereby $w$ is oriented towards $R$ since $R$ is $x$-accessible. Otherwise if $R=\{u,l\}$, then the tree-root of $ST(G+(x,S))$ is either $u$ (if $u$ is a clique-node) or $\{u,w\}$ (otherwise). It follows that in every cases all the star-nodes of $ST(G+(x,S))$ are oriented towards $R$.
\end{itemize}
\end{proof}

Observe that, unlike in the vertex incremental characterization of DH graph (see Theorem~\ref{th:charac}), Theorem~\ref{th:inc-cograph} does not require the restriction $lS|>1$. This case is indeed captured by condition 2 on $T(S)$.
%------------------------------------------------------------------------------------------------------------------
\subsection{$3$-leaf powers}

\begin{definition}
For an integer $k$, a graph $G=(V,E)$ is a \emph{$k$-leaf power} if there is a tree $T$ whose leaf set is $V$ and such that $xy\in E$ if and only if the distance in $T$ between leaves $x$ and $y$ is at most $k$, $d_T(x,y)\leqslant k$. The tree $T$ is called \emph{root-tree} of $G$.
\end{definition}

The family of $k$-leaf power has been introduced in~\cite{NRT02} in the context of phylogenetic tree reconstruction. Forbidden induced subgraph characterizations are known for $k\leqslant 4$.
In~\cite{BL06}, $3$-leaf powers have been characterized as the graphs resulting from the substitution of vertices of a tree by cliques. This leads to the following alternative characterization (see Figure \ref{fig:3-leaf}).

\begin{figure}[htbh]
\inserfig{0.9}{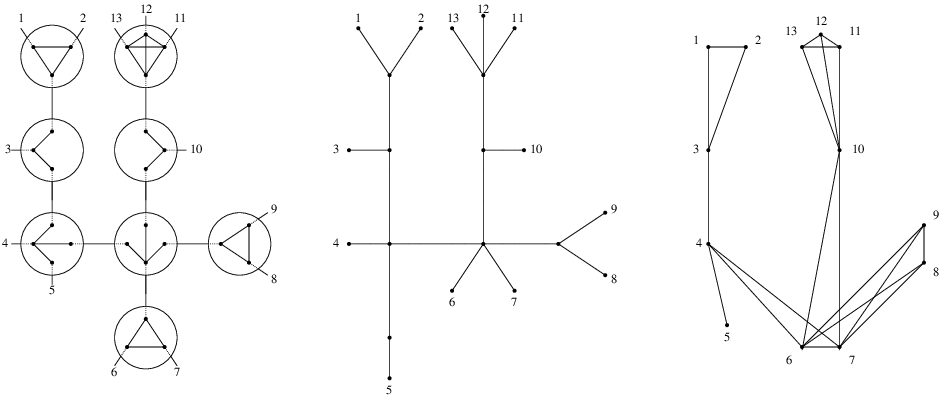}
%\centerline{\includegraphics[width=0.7\linewidth]{3-leaf.pdf}}
\caption{A $3$-leaf power graph on the right, its split tree on the left, and a root-tree of this graph in the middle.}
\label{fig:3-leaf}
\end{figure}

\begin{theorem}[$3$-leaf power split tree characterization] \label{th:3-leaf}
A connected graph $G=(V,E)$ is a $3$-leaf power if and only if 
\begin{enumerate}
\item its split tree $ST(G)=(T,\mathcal{F})$ is a clique-star tree ($G$ is distance hereditary);
\item the set of star nodes forms a connected subtree of $T$; 
\item if $u$ is a star node, then the tree-edge $e$ such that $\rho_u(e)$ is the centre of the star, is incident to a leaf or a clique node.
\end{enumerate}
\end{theorem}

\begin{proof}
We assume that $G$ is not a clique nor a star, otherwise the statement is trivially true.
\begin{itemize}
\item[$\Rightarrow$] 
As $G$ is a $3$-leaf power there exists a root-tree $T'$ whose leaf set is $V$. Assume first that no pair of leaves are at distance two in $T'$.  For a leaf $x$, we denote by $n(x)$ its unique neighbor. Clearly $x$ and $y$ are adjacent in $G$ if and only if $n(x)$ and $n(y)$ are adjacent in $T'$. As $G$ is connected, every node of $T'$ is the neighbor of some leaf.  Let us construct a graph-labelled tree $(T',\mathcal{F}')$ such that $\G(T',\mathcal{F}')=G$. The graph label $G_v$ of each node $v=n(x)$ is a star whose centre is $\rho_v(xv)$. It is clear that two leaves of $(T',\mathcal{F}')$ are adjacent in $\G(T',\mathcal{F}')$ if and only if there are attached to the centre of two neighboring stars in $T'$: i.e. $\G(T',\mathcal{F}')=G$. As no pair of leaves are at distance two, $T'$ may contain some node of degree $2$. Then performing a node-join on each such node $u$ and its non-leaf neighbor $v$, yiedls a graph-labelled tree $(T,\mathcal{F})$ which is reduced and which only contains stars: this is the split tree $ST(G)$.

Now assume that $T'$ contains some pairs of leaves at distance $2$. Such a pair of leaves defines a pair of true twins in $G$. Let $\mathcal{P}$ be that partition of  $V(G)$ (leaf set of $T'$) into maximal sets of true twins (or maximal clique modules). The split tree of the quotient graph $G{/\mathcal{P}}$
is obtained as described above. Now the clique modules are reintroduced by performing true twins insertions (see Figure~\ref{fig:incremental}) in the split tree. Let $x$ be a leaf of $ST(G{/\mathcal{P}})$ and $M_x$ be the corresponding clique module. Then subdivide the tree-edge incident to $x$ by a clique node of degree $1+|M_x|$ (see Figure~\ref{fig:incremental} for a true twins augmentation). This yields a split tree of $G$ having the expected properties.

\item[$\Leftarrow$] Assume that $ST(G)=(T,\mathcal{F})$ satisfies conditions 1, 2 and 3. Then the root-tree $T'$ whose leaf set is $V$ (i.e. equal to the leaf set of $T$) is obtained as follows: 1) contract every tree-edge $uv$ of $T$ such that $u$ is a clique node and $v$ is a star node; and 2) subdivide every tree-edge $e=vl$ of $T$ 
such that $l$ is a leaf, $v$ is a star node and $\rho_v(e)$ is not the centre of the star $G_u$. Let us prove the correctness of this construction.

Assume first that $ST(G)$ only contains star nodes. Let $l$ be a leaf and $u$ be its neighbor.
Suppose that $\rho_u(lu)$ is not the centre of the star $G_u$. As $e=lu$ is a subdivided tree-edge, $d_{T'}(l,l')\geqslant 3$ with every leaf $l'\neq l$. In this case no contraction is performed, and thereby the distances between leaves do not decrease. Observe then that the only leaf $l'$ such that $d_{T'}(l,l')=3$ is attached to the centre of the star $G_u$ (i.e. $\rho_u(l'u)$). It is clear that $l'$ is the only leaf accessible to $l$ in $ST(G)$, i.e. adjacent in $G$. So suppose that $\rho_u(lu)$ is the centre of the star $G_u$. As just argued, $d_{T'}(l,l')=3$ for every leaf $l'\neq l$ adjacent to $u$ and $l$, $l'$ are pairwise accessible in $ST(G)$ so adjacent in $G$. So consider a leaf $l'$ adjacent to a node $v$ distinct from $u$. Observe that if $u$ and $v$ are not adjacent, then $d_{T'}(l,l')>3$ and by condition 3 $l$ cannot be accessible from $l'$. Otherwise ($u$ and $v$ are adjacent nodes), if $\rho_v(l'v)$ is the centre of $G_v$ then $d_{T'}(l,l')=d_T(l,l')=3$ which is fine since $l$ is accessible from $l'$. If $\rho_v(l'v)$ is not the centre of $G_v$ then $d_{T'}(l,l')=d_T(l,l')+1=4$ but then $l$ is not accessible from $l'$. It follows that $l$ and $l'$ are at distance $3$ is $T'$ if and only if there are adjacent in $G$.

To conclude consider the case where $ST(G)$ contains some clique nodes. Observe that by condition 2, a clique node $u$ is adjacent to at most one star node. Observe also that every pair of leaves adjacent to the same clique node are (adjacent) twins. Now if we save only one representative leaf per clique node, we obtain a graph $G'$ whose split tree $ST(G')$ only contains star nodes (replace every clique node by the corresponding representative leaf). We have shown that our root-tree construction is valid for $G'$. By the observations above, to obtain the root-tree it suffices to add every non-representative leaf $l$ adjacent to the same node than its representative. Observe that this finally amount to contract the tree-edge between clique-nodes and star-nodes. This conclude the proof.
\end{itemize}
\end{proof}

Observe that, to get a 3-leaf power vertex incremental characterization, we could simply test, given a 3-leaf power graph $G$, first if the graph $G+x$ is a DH graph using Theorem~\ref{th:charac}, and then if the node to which $x$ is attached in $ST(G+x)$ does not create a contradiction with Theorem~\ref{th:3-leaf}. This second condition amounts to test a local condition on $ST(G+x)$, and would be enough
for algorithmic purpose to refine the main DH algorithm of Section  \ref{sec:vertex-algo} in terms of 3-leaf power graphs as done in Section \ref{sub:vertex-3leaf}.
However, the following theorem establishes a more precise property directly on $ST(G)$.

\begin{theorem}[$3$-leaf power vertex incremental characterization]\label{th:3-leaf-inc}
Let $G$ be a connected \break
\hbox{$3$-leaf 
power} and $ST(G)=(T,\mathcal{F})$ be its split tree. Then $G+(x,S)$ is a $3$-leaf power if and only if
\begin{enumerate}
\item it is a distance hereditary graph (see conditions of Theorem~\ref{th:charac});

\item if $S=\{y\}$, then either $y$ is adjacent in $T$ to a star node, or $T$ has a only one node;

\item if $|S|>1$, then
\begin{enumerate}
\item if $T(S)$ does not contain a partially accessible node, then the tree-edge, towards which the fully-mixed star nodes are oriented (see Theorem~\ref{th:charac}), is incident to a clique node or a leaf;

\item if $T(S)$ contains a partially accessible node $u$, then $u$ is a clique node, and either $S$ is the set of leaves adjacent to $u$ or $u$ is the only node of $ST(G)$.
\end{enumerate}
\end{enumerate}
\end{theorem}

\begin{proof}

We first consider the case $S=\{y\}$. Then $ST(G+(x,S))$ is obtained from $ST(G)$ by inserting on the tree-edge incident to $y$ a degree $3$ star node $u$ adjacent to $x$ and whose  centre is $\rho_u(uy)$. Thanks to Theorem~\ref{th:3-leaf}, $G+(x,S)$ is a $3$-leaf power if and only if the neighbor of $y$ if condition 2 is satisfied.

From now on, we assume that $|S|>1$ and prove that $G+(x,S)$ is DH if and only if conditions 1 and 3 hold.
\begin{itemize}
\item[$\Rightarrow$] Let us consider the three different ways $ST(G)$ can be transformed into $ST(G+(x,S))$:
\begin{enumerate}
\item \textit{Vertex $x$ is attached to a node $u$ of $ST(G)$.}
Assume that $u$ is a clique node. Then,  by Corollary~\ref{cor:neighbor-node}, $u$ is a fully accessible clique node of $T(S)$ and $T(S)$ does not contain any partially accessible node. It follows that every star node of $T(S)$ is oriented towards any tree-edge of $T(S)$ incident to $u$. Consider the case $u$ is a star node. Then by Theorem~\ref{th:3-leaf} and since $|S|>1$, the neighbor $v$ of $u$, such that $\rho_u(uv)$ is the centre of $G_v$, is a clique. It follows that  $v$ is the partially accessible node of $T(S)$ and $S$ is the set of leaves adjacent to $v$.

\item \textit{A node $u$ of $ST(G)$ is node-split into two adjacent nodes $v$ and $v'$ and the tree-edge $vv'$ is subdivided to insert a new node $w$ adjacent to $x$.} As observed in the proof of Theorem~\ref{th:charac}, $u$ is partially accessible (this is a consequence of Corollary~\ref{cor:neighbor-node}). Assume that $u$ is a star node. Then, by Theorem~\ref{th:3-leaf}, $w$ cannot be a clique node, since otherwise it would neighbor two star nodes, namely $v$ and $v'$. But if $w$ is a star node, then the tree-edge $e$ such that $\rho_u(e)$ is the centre of $G_u$ is adjacent to a star node $v\not=u$: contradicting Theorem~\ref{th:3-leaf} again. It follows that $u$ has to be a clique node. This forces $w$ to be a star node. Theorem~\ref{th:3-leaf} then implies that $G+(x,S)$ is a 3-leaf power graph if and only if $u$ is the unique node of $ST(G)$ (otherwise the set of star nodes in $ST(G+(x,S))$ would not be connected). 

\item \textit{A tree-edge $e$ of $ST(G)$ is subdivided to insert a new node $w$ adjacent to $x$.} 
If $w$ is a clique node, then by Corollary~\ref{cor:neighbor-node}, $T(S)$ does not contain any partially accessible node. By Theorem~\ref{th:3-leaf}, $G+x$ is a 3-leaf power graph if and only if $e$ is incident to a leaf of $ST(G)$. 
Assume that $w$ is a star node with centre $\rho_w(vw)$. As $|S|>1$, $v$ is not a leaf. By Corollary~\ref{cor:neighbor-node}, $v$ is partially accessible. By Theorem~\ref{th:3-leaf}, $G+x$ is a 3-leaf power graph if and only if $v$ is a clique node. Moreover in that case, observe that $S$ is precisely the set of leaves of $T$ adjacent to $v$.
\end{enumerate}

\item[$\Leftarrow$] We just observe that if condititons (3.a) and (3,b) hold, then the construction of $ST(G+(x,S))$ described in the proof of Theorem~\ref{th:charac} yields a split tree that satisfies Theorem~\ref{th:3-leaf}. We describe the two cases more precisely. Assume condition (3.a) holds. Let $e$ be the tree-edge of $T(S)$ towards which the fully-mixed star nodes are oriented. Then either $e$ is incident to a leaf, or $e$ is incident to a star and a clique. In both, cases, the construction of $ST(G+(x,S))$ from $ST(G)$ described in the proof of Theorem~\ref{th:charac} shows that $ST(G+(x,S))$ satisfies the conditions of Theorem~\ref{th:3-leaf}. Assume now that $T(S)$ contains a partially accessible node and condition (3.b) holds. Again from the proof of Theorem~\ref{th:charac}, we know that to get $ST(G+(x,S))$ from $ST(G)$, the partially accessible node $u$ is node-split. Since $u$ is a clique node, it is then straightforward to check that condition (3.b) implies that $ST(G+(x,S))$ satisfies the conditions of Theorem~\ref{th:3-leaf}. 
\end{itemize}
\end{proof}

%------------------------------------------------------------------------------------------------------------------
%------------------------------------------------------------------------------------------------------------------
\section{Vertex-only fully-dynamic recognition algorithms}
\label{sec:vertex-algo}

The main result presented in this section is an optimal vertex-only fully dynamic algorithm that maintains the split tree representation of a DH graph. For both insertion and deletion queries, the split tree can be updated in time $O(d(x))$, where $d(x)$ is the degree of the vertex to be inserted or deleted. In the case of an insertion, the algorithm can check whether the resulting graph is DH or not. 
As corollaries, we obtain linear time recognition and isomorphism algorithms for DH graphs.
We also give a short overview of how this algorithm can be specialized for the cases of cographs and of $3$-leaf powers.

Let us first describe the data-structure we use to implement the split tree of the input graph.

\paragraph{Data-structure.}
The following data structure is used to encode the clique-star tree $ST(G)=(T,\mathcal{F})$ of the given connected DH graph $G$:
\begin{enumerate}
\item  a (rooted) representation of the tree $T$. The root of $T$ is chosen arbitrarily and is only required for the seek of computational efficiency;
\item as the graphs of $\mathcal{F}$ are cliques or stars, each node of $T$ only needs a \emph{clique-star mark} distinguishing the type of each node, the degree of the node and in the case of a star  a \emph{centre mark} to distinguish its centre from the other marker-vertices;
\end{enumerate}
Such a data structure is clearly an $O(n)$ space representation of any DH graph on $n$ vertices.

%------------------------------------------------------------------------------------------------------------------
\subsection{Vertex-insertion in DH graphs}
\label{sub:vertex-DH}

The insertion algorithm works in three steps. Given a DH graph $G$ represented by its split tree $ST(G)$ and a new vertex $x$ together with a set of vertices $S$ of $G$: 1) we first compute the subtree $T(S)$; 2) then we check whether the conditions of Theorem~\ref{th:charac} are satisfied; and finally 3) if the augmented graph $G+x$ turns out to be DH, we update the split tree data-structure (otherwise the algorithm stops).

%--------------------------------------------
\paragraph{\textbf{Computing the smallest subtree spanning a set of leaves.}}

Given a set $S$ of leaves of a tree $T$, we need to identify the smallest subtree $T(S)$ spanning $S$, and to store the degrees of its nodes. This problem is easy to solve on rooted trees by a  bottom-up marking process in time $O(|T(S)|)$ as follows:

\begin{enumerate}
\item Mark each leaf of $S$. Along the algorithm, a marked node is \emph{active} if it is not the root and its father is not marked.
\item Each active node marks its father if: 1) the root is not marked and there are at least two active vertices, or 2) the root is marked and there is at least one active node.
\item While the root of the subtree $T'$ induced by the marked nodes is a leaf of $T'$ but does not belong to $S$, then remove this (root) node from $T'$, let its child be the new root of $T'$ and check again. Eventually return $T(S)=T'$.
\end{enumerate}

By Lemma~\ref{lem:tree-size}, if the augmented graph $G+x$ is DH, the size of $T(S)$ (its number of nodes) is at most $2.|S|$. To prevent a non-linear complexity in the case $G+x$ is not DH, while computing $T(S)$, we need to count the number of marked nodes. More precisely after step 2,
the number of marked nodes is at most $2.|T(S)|$ (since the number of deleted nodes in step 3 cannot exceed the number of marked nodes). Hence if the graph is DH, this number of marked nodes is at most $4.|S|$. Whenever more than $4.|S|$ nodes have been marked during step 2, the algorithm stops and claims that the graph $G+x$ is not DH.
In every case, it is easy to check that the above algorithm has  $O(|S|)$ running time. Its correctness is straightforward.

%--------------------------------------------------------------------------
\paragraph{\textbf{Testing conditions of Theorem~\ref{th:charac}.}}

The first two conditions of Theorem~\ref{th:charac} are fairly easy to check by following Definition~\ref{def:accessible}: a node $u$ is \emph{fully accessible} if its degrees in $T(S)$ and $T$ are the same; 
$u$ is \emph{singly accessible} if it is a star, if it has degree $2$ in $T(S)$ and if the neighbor $v$ of $u$, such that $\rho_u(uv)$ is the star centre, belongs to $T(S)$; and $u$ is \emph{partially accessible} otherwise (such a node has to be unique if it exists). These tests cost $O(|T(S)|)$.

We can now assume that the first two conditions of Theorem~\ref{th:charac} are fulfilled. Since the case $|S|=1$ is trivial, we also assume that $|S|>1$.

We define \emph{local orientations} on nodes of a tree as the choice,  for each node $u$, of a node $f(u)$ such that either $f(u)=u$ or $f(u)$ is a neighbor of $u$. Local orientations are \emph{compatible} if 1) $f(u)=u$ implies $f(v)=u$ for every neighbor $v$ of $u$, and 2) $f(u)=v$ implies $f(w)=u$ for every neighbor $w\not=v$ of $u$.
An easy exercise is to see that if local orientations are compatible then exactly one of the two following properties holds: either there exists a unique node $u$ with $f(u)=u$, in which case $u$ is called \emph{node-root}, or there exists a unique tree-edge $uv$ with $f(u)=v$ and $f(v)=u$, in which case $uv$ is called~\emph{tree-edge-root}.

Testing the third condition of Theorem~\ref{th:charac} consists of building, if possible,
compatible local orientations in the subtree $T(S)$:
\begin{enumerate}
\item Let $u$ be a leaf of $T(S)$. Then $f(u)$ is the unique neighbor of $u$. 

\item Let $u$ be a star node of $T(S)$.  If $u$ is partially accessible, then $f(u)=u$.
If $u$ is singly accessible, then $f(u)$ is the unique neighbor $v$ of $u$ belonging to $T(S)$ such that $\rho_u(uv)$ is a degree-1 vertex of the star. 
If $u$ is fully accessible, then $f(u)$ is the neighbor $v$ of $u$ such that $\rho_u(uv)$ is the centre of the star. 

\item Let $u$ be a clique node of $T(S)$. If $u$ is partially accessible, then $f(u)=u$. Otherwise, $u$ is fully accessible and its neighbors are leaves or star nodes.  If $f(v)=u$ for every neighbor $v$ of $u$ then $f(u)=u$. 
If $f(v)=u$ for every neighbor $v$ of $u$ but one, say $w$, then $f(u)=w$. Otherwise $u$ is an \emph{obstruction}. 
\end{enumerate}

The third condition of Theorem \ref{th:charac} is satisfied if and only if 1)  there is no obstruction and 2) local orientations of $T(S)$ are compatible. This test can be performed in time $O(|T(S)|)$ by a search of $T(S)$.
Hence the conditions of Theorem \ref{th:charac} can be tested in $O(|T(S)|)$ time. 
Moreover if the test is satisfied, the search of $T(S)$ locates the node-root or the tree-edge-root.

%--------------------------------------------------------------------------
\paragraph{\textbf{Updating the split tree.}}

We now assume that $G+(x,S)$ is DH (i.e. conditions of Theorem~\ref{th:charac} are satisfied).
So by Theorem~\ref{th:charac} the split tree has either a unique node-root or a  unique tree-edge-root.
To update the split tree, we may subdivide an insertion tree-edge into two new tree-edges.
Notice that, since we maintain an (artificial) orientation of the tree, this subdivision can be done in $O(1)$.
There are three~cases to consider (see~Figure~\ref{fig:insertion_cases}), after a possible single node-split preprocess (see~Figure~\ref{fig:insertion_preprocess}).

\begin{enumerate}

\setcounter{enumi}{-1} 

\item \emph{Single node-split preprocess:} If there is a node-root $u$ being partially accessible, then, depending on degree conditions on $u$, a preliminary update of  $T$ consisting of a node-split of the node $u$ is required. Let $U$, resp. $A$, be the set of tree-edges incident to $u$  in $T$, resp. in $T(S)$. 

\begin{figure}[h] 
\inserfig{1}{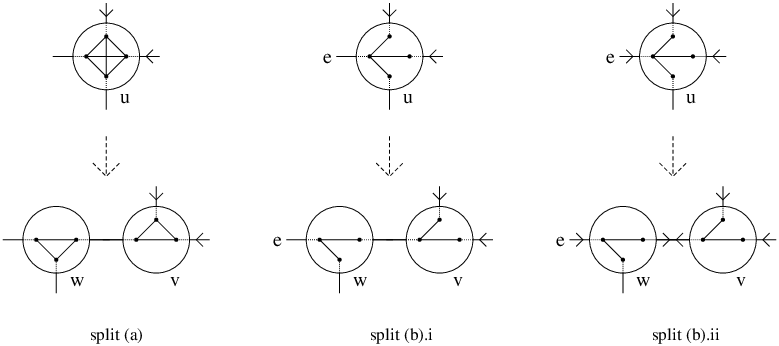}
%\centerline{\includegraphics[width=1\linewidth]{insertion.eps}}
\caption{Vertex-insertion preprocessing step: a node-split on the node-root $u$ is requited to separate the set $A$ of tree-edges (i.e. those incident to $u$ and belonging to $T(S)$ - drawn with an arrow in the figure) from the others.} 
\label{fig:insertion_preprocess}
\end{figure}

\begin{enumerate}
\item If $u$ is a clique node with $|U\setminus A|\geq 2$, then $u$ is node-split. Two new adjacent clique nodes $v$ and $w$ are created in $T$. The marker-vertices of $v$ (resp. $w$) correspond to $A$, resp. $U\setminus A$,  except one which corresponds to $vw$.
In this case, $v$ is now the (partially accessible) node-root.

\item If $u$ is a star node, the centre of which is mapped to the tree-edge $e$, and $|(U\setminus A)\setminus e)|\geq 1$,
then $u$ is node-split and replaced by two adjacent star nodes $v$ and $w$.
Then the extremities of the star $G_v$ correspond to $A\setminus\{e\}$ and its centre to $vw$ (we have $|A\setminus\{e\}|>1$ since $u$ is not singly accessible), likewise the extremities of the star $G_w$ correspond to $(U\setminus A)\cup\{vw\}$ and its centre  to $e$.

\begin{enumerate}

\item If $e\not\in A$, then the node $v$ becomes the (partially accessible) node-root.

\item If $e\in A$, then the tree-edge $vw$ is now the tree-edge-root.

\end{enumerate}
\end{enumerate}

\item \emph{The root of $T(S)$ is a partially accessible node $v$, or $S$ is reduced to a unique leaf $v$.}
Let $w$ be its neighbor in $T$ that does not belong to $T(S)$. 
Then the insertion tree-edge is $e=vw$, and $ST(G+(x,S))$ is obtained by subdividing $vw$ into two tree-edge $vr$ and $rw$, where $r$ a degree $3$ star node whose centre is $\rho_w(vr)$ and to which $x$ is adjacent. 
Finally if $w$ is a star with centre $\rho_v(wr)$, we proceed a node-join operation on the tree-edge $wr$ .

\item \emph{The root of $T(S)$ is a node $v$ which is not partially accessible.}
By the definition of the local orientation $f$, the node $v$ is a clique node, and $ST(G+(x,S))$ is obtained by adding the new leaf $x$ adjacent to $v$ whose degree thereby increases by one.

\item \emph{The root of $T(S)$ is a tree-edge $vw$.}
Then $ST(G+(x,S))$ is obtained by subdividing $vw$ with a clique node $r$ of degree $3$ and making the leaf $x$ adjacent to $r$. 
\end{enumerate}

\begin{figure}[h] 
\inserfig{1}{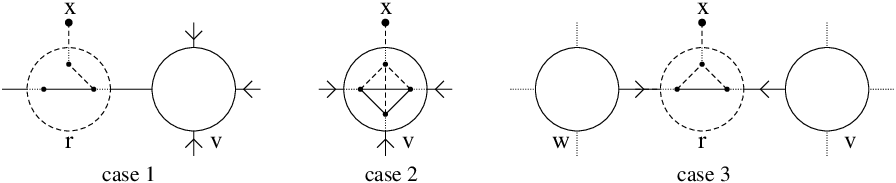}
%\centerline{\includegraphics[width=1\linewidth]{insertion.eps}}
\caption{The three different cases for the vertex-insertion: 1) the root of $T(S)$ is a partially accessible node $v$; 2) the root of $T(S)$ is a node $v$ which is not partially accessible; and 3) the root of $T(S)$ is a tree-edge $vw$. The modified split tree is obtained by inserting dashed node or tree-edges.} \label{fig:insertion_cases}
\end{figure}

\begin{theorem}[Vertex insertion]
\label{th:vins}
Let $G+(x,S)$ be a graph such that $G$ is a connected distance hereditary graph. Given the data structure of the split tree $ST(G)$,
testing whether $G+(x,S)$ is distance hereditary and if so computing  the data structure of $ST(G+(x,S))$ can be done in $O(|S|)$~time.
\end{theorem}

\begin{proof}
The correctness follows from the discussion above and the proof of Theorem~\ref{th:charac}.

Concerning the complexity issues, every tree modification operation can be done in $O(1)$ time, except the splitting in case 0 which requires $O(|T(S)|)$ time (by deleting $A$ from $u$ to get $w$, and adding $A$ to a new empty node $v$).
Any other operation time to maintain the data structure of the split tree (root, degrees...) requires $O(1)$ time. Then, the complexity for the whole insertion algorithm derives from previous steps and the fact that $O(|T(S)|)=O(|S|)$ if the algorithm has passed the $T(S)$ computation  step.
\end{proof}

\bigskip
Let us remark that our vertex-insertion algorithm yields a linear time recognition algorithm of (static) DH graphs, thereby achieving the best known bound but also simplifying the previous non-incremental ones~\cite{HM90,DHP01,Bre05}. It also yields a linear time isomorphism algorithm, thereby achieving the best known bound again with a simpler setting than in \cite{NUU07}.

\begin{corollary}[Static recognition] \label{cor:static}
The vertex-insertion routine enables to recognize distance hereditary graphs in linear time.
\end{corollary}

\begin{proof}
As the insertion algorithm works only on connected graphs, we have to proceed the vertices in an ordering $x_1,\dots x_n$ such that, for every $1\leqslant i\leqslant n$, $G[\{x_1,\dots x_i\}]$ is connected. Any search (e.g. BFS) computes such an ordering in linear time. As the global complexity cost is linear in the sum of the degrees, linear time follows.
\end{proof}

\begin{corollary}[Isomorphism] \label{cor:iso}
The vertex-insertion routine enables to test distance hereditary graph isomorphism in linear time.
\end{corollary}
\begin{proof}
To test isomorphism between two DH graphs, it suffices to test isomorphism between the two corresponding split trees. The split tree of a DH graph can be constructed in linear time by our recognition algorithm and has size linear in the number of vertices of the graph (Lemma \ref{lem:tree-size}). Thereby any linear time tree isomorphism algorithm can be used (e.g. \cite{AHU74}).
\end{proof}

%------------------------------------------------------------------------------------------------------------------
\subsection{Vertex-deletion in DH graphs}
\label{sub:vertex-del-DH}

Removing a vertex $x$ from a DH graph $G$ always yields a DH graph $G-x$. 
Let $ST(G)$ be the split tree of $G$.
Updating the data structure of the split tree can be done as follows.

\begin{enumerate}
\item Remove the leaf $x$ and update the degree of its neighbor $v$.
\item If $v$ now has degree $2$, then remove $v$ and add a tree-edge between its neighbors $u$ and $w$. If the resulting clique-star tree is not reduced, proceed a node-join on the tree-edge $uw$.
\item If $v$ is a star node whose centre neighbor was $x$, then $G-x$ is no longer connected, and the split-trees of each connected component are the components of $T-\{v,x\}$.
\end{enumerate}

\begin{lemma}[Vertex deletion]
\label{th:vdel}
Let $G$ be a connected distance hereditary graph and $x$ be a degree $d$ vertex of $G$.
Given the data structure of split tree $ST(G)$,
testing whether $G-x$ is a connected distance hereditary graph and if so computing  the data structure of $ST(G-x)$ can be done in $O(d)$ time.
\end{lemma}
\begin{proof}
Every operation, except the node-join, can be achieved in $O(1)$ time. The complexity of the node-join on the tree-edge $uw$ is $min(d(u),d(w))$, where $d(u),d(w)$ are respectively the degree of node $u$ and node $w$.
Since at least one of these nodes is fully accessible, this minimum degree is smaller than $d$, the degree of $x$.
Hence this node-join operation costs $O(d)$.
\end{proof}
\bigskip

To summarize the results of vertex dynamic DH graphs, with Theorem \ref{th:vins} and Lemma \ref{th:vdel}, we have proved that:

\begin{theorem}
There exists a vertex fully dynamic recognition algorithm for connected distance hereditary graphs, maintaining  the split tree, with complexity $O(d)$ per vertex-insertion or deletion operation involving $d$ edges.
\end{theorem}

%------------------------------------------------------------------------------------------------------------------
\subsection{Vertex modifications in cographs}
\label{sub:vertex-cograph}

To check whether the augmented graph $G+(x,S)$ is a cograph, our vertex-insertion algorithm for DH could be used. According to Theorem~\ref{th:inc-cograph}, we just need an extra test to verify that the tree-root has a node in the subtree $T(S)$ or is neighboring a node of $T(S)$. Notice that as the original graph $G$ is a cograph, the star nodes define a natural orientation which can be used to  compute $T(S)$. Let us also remark that, as a consequence of Theorem~\ref{th:inc-cograph}, the set of singly accessible nodes (which are stars) has to belong to a path from the tree-root of $ST(G)$ to some node $u$. It follows that to test condition 3 of Theorem~\ref{th:charac}, the local orientations can be avoided. This path property for the singly accessible nodes was already noticed (in other terms) in the characterization proposed in~\cite{CPS85}. Finally, we need an extra work to update the tree-root as described in the proof of Theorem~\ref{th:charact-cograph}. This can also be done in constant time. It follows that the resulting complexity is $O(d)$ by insertion as in the incremental recognition algorithm of Corneil, Perl and Stewart~\cite{CPS85} (which is based on the modular decomposition tree).

As cographs are hereditary graphs, the vertex-deletion always yields a cograph. Notice also that removing a vertex does not affect the orientation of the remaining star-nodes in the split tree. It follows that our vertex-deletion algorithm for DH graph can be used as well for the vertex-deletion of cographs.

\begin{theorem} 
There exists a  vertex fully dynamic recognition algorithm for connected cographs, maintaining  the split tree, with complexity $O(d)$ per vertex-insertion or vertex-deletion operation involving $d$ edges.
\end{theorem}

%------------------------------------------------------------------------------------------------------------------
\subsection{Vertex modifications in $3$-leaf powers}
\label{sub:vertex-3leaf}

Again the DH vertex-insertion algorithm can be easily specialized to  work on $3$-leaf powers. 
%EME-19-12-10 tu as note "`NO"' mais j'ai pas compris
Thanks to Theorem~\ref{th:3-leaf-inc}, insertion of a pendant vertex $x$ neighboring $y$ is restricted to the case where a leaf $y$ is adjacent to a star node or the split tree has a unique node. This can be checked in $O(1)$ time. In the other cases, we just need to test whether the subtree $T(S)$ contains or not a partially accessible node. This only requires a search of $T(S)$ whose size is $O(|S|)$. Concerning the deletion algorithm, as $3$-leaf powers are hereditary graphs, we just apply the DH vertex-deletion algorithm.

\begin{theorem}
There exists a  vertex fully dynamic recognition algorithm for connected $3$-leaf powers, maintaining  the split tree, with complexity $O(d)$ per vertex-insertion or vertex-deletion operation involving $d$ edges.
\end{theorem}

Notice that since the family of $3$-leaf power is hereditary, this vertex incremental recognition algorithm also applies to static graph. The time complexity is linear as for the recognition algorithm proposed in~\cite{BL06}. Moreover our algorithm can be easily adapted to output the tree root when the input graph is a $3$-leaf power.

%------------------------------------------------------------------------------------------------------------------
%------------------------------------------------------------------------------------------------------------------
\section{Edge modifications: characterizations and algorithms.}

In this section we show that the split tree representation is also the right tool to deal with edge modifications in totally decomposable graphs. Indeed, based on the forbidden induced subgraph characterizations of the three graph families we have considered so far (DH graphs, cographs and $3$-leaf powers), we identify necessary and sufficient conditions under which  given a graph $G$ and an edge $e$, the modified graph $G+e$ (or $G-e$) belongs to the same family than $G$. Using the graph-labelled tree representation, these conditions consist in checking if a given path in the split tree belongs to a small finite set of configurations. These simple characterizations yield to simple constant time edge fully-dynamic algorithms. Let us mention that such algorithmic results were already known for cographs~\cite{SS04} and DH graphs~\cite{TC07}. For cographs, the edge fully-dynamic algorithm in~\cite{SS04} relies on a modular decomposition based characterization which, again, we are able to transpose in the split decomposition settings, and which are derived as a particular case of the DH edge modification algorithm. Concerning the DH graphs, the constant time algorithm of~\cite{TC07} is way more complicated than the one we propose here. It relies on a tricky analysis on the BFS layering structure~\cite{HM90} of DH graphs and up to our knowledge no simple characterization could be identified from that work. 
No result of this flavour was known for $3$-leaf powers.

%------------------------------------------------------------------------------------------------------------------
\subsection{Edge-modification in distance hereditary graphs.}
\label{sub:edge-DH}

This subsection states our results on edge modifications in DH graphs.
The combinatorial characterization Theorem \ref{th:edge_dynamic}
directly  implies the algorithm of Corollary \ref{th:cor_edge_dynamic}
and is proved in the next subsection.

Let $G$ be a connected DH graph and $ST(G)=(T,\mathcal{F})$ be its split tree. If $x$ and $y$ are two vertices of $G$, we denote $P(x,y)$ the graph labelled tree formed by the path in $T$ between the leaf $x$ and the leaf $y$, with nodes labelled the same way as in $ST(G)$. 
As $ST(G)$ is a clique-star tree, $P(x,y)$ naturally defines a word $W(x,y)$ whose letters identify the type of the graphs labelling the nodes in $P(x,y)$. 
A alphabet of four symbols $A=\{K,S,S_x,S_y\}$ is enough to describe $W(x,y)$:
\begin{itemize}
\item the letter $K$ stands for the clique nodes;
\item the letter $S$ stands for the star nodes $v$, the centre $\rho_v(e)$ of which is mapped to the tree-edge $e$ that does not  belong to $P(x,y)$; and
\item the letter $S_x$ (resp. $S_y)$ stands for the star nodes $v$, the centre $\rho_v(e)$ of which is mapped to the tree-edge $e$ that belongs to the subpath of $P(x,y)-v$ containing $x$ (resp. $y$).
\end{itemize}

Observe that $xy\in E(G)$ if and only if $W(x,y)$ is $S$-free (i.e. does not contain the letter $S$). 
When describing words, letters in  brackets can be deleted: e.g. $K(S)K$ stands for the words $KK$ and $KSK$.

\begin{theorem} \label{th:edge_dynamic}
Let $G$ be a connected DH graph and $ST(G)=(T,\mathcal{F})$ be its split tree. Let $x$ and $y$ be two vertices of $G$ and $W(x,y)$ be the word labelling the path $P(x,y)$ between $x$ and $y$ in $T$. Then
\begin{enumerate}
\item If $xy\notin E$, then $G+xy$ is distance hereditary if and only if $W(x,y)$ is one of the following words:\\ 
\centerline{$(S_x)SS(S_y)$ \hfil $(S_x)SK(S_y)$ \hfil $(S_x)KS(S_y)$ \hfil $(S_x)S(S_y)$}
\item If $xy\notin E$, then $G-xy$ is distance hereditary if and only if $W(x,y)$ is one of the following words:\\
\centerline{$(S_x)S_yS_x(S_y)$ \hfil $(S_x)S_yK(S_y)$ \hfil $(S_x)KS_x(S_y)$ \hfil $(S_x)K(S_y)$ \hfil $(S_x)(S_y)$}

Moreover if $W(x,y)=(S_x)(S_y)$, then $G-xy$ is no longer connected.
\end{enumerate}
\end{theorem}

\begin{corollary} \label{th:cor_edge_dynamic}

The following algorithm tests and updates the data-structure of the split tree for the insertion or deletion of an edge $xy$ in a (connected) distance hereditary graph $G$ in constant time.
\end{corollary}

\begin{enumerate}

\item Test if $W(x,y)$ has length at most $4$ and satisfies conditions of Theorem~\ref{th:edge_dynamic}. 

\item Update the split tree of $G$. Nodes of letters in  brackets are called~{\it extreme}.%

\begin{enumerate}

\item Node-split every non-extreme node of $W(x,y)$ that is not ternary so that in the resulting clique-star tree, all the non-extreme node of $W(x,y)$ are ternary.

\item Replace the non-extreme nodes by ternary nodes according to the following table. If $W(x,y)$ contains two non-extreme nodes, say $u$ and $v$, then the neighbor $u'$ of $u$ (resp. $v'$ of $v$), that does not belong to $W(x,y)$, becomes adjacent to $v$ (resp. $u$). See Figure \ref{fig:edge_algo}. Extreme nodes are left unchanged.
\medskip

{
\noindent
\hfill{
\begin{tabular}{|c|c|}
\hline
\multicolumn{2}{| c |}{
edge-insertion $\longrightarrow$
}\\
\multicolumn{2}{| c |}{
$\longleftarrow$ edge-deletion
}\\
\hline

$\ \ \ \ (S_x)SS(S_y)\ \ \ \ $				& $\ \ \ \ (S_x)S_yS_x(S_y)\ \ \ \ $ \\ 
$(S_x)SK(S_y)$				& $(S_x)S_yK(S_y)$ \\ 
$(S_x)KS(S_y)$				& $(S_x)KS_x(S_y)$ \\ 
$(S_x)S(S_y)$				& $(S_x)K(S_y)$ \\ 

\hline
\end{tabular}
}\hfill
}

\medskip

\item If necessary, proceed (at most two) node-join operations involving the nodes that have been changed to get a reduced graph-labelled~tree.

\end{enumerate}

\end{enumerate}

\begin{figure}[h]
\inserfig{0.85}{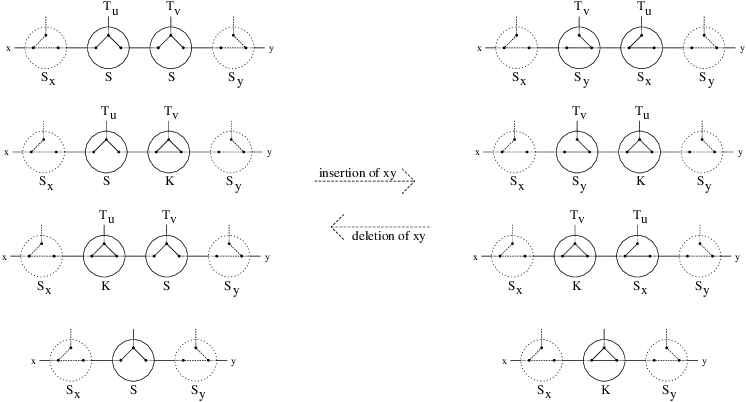}
\caption{Constant time dynamic algorithm for edge modification in DH graphs (Corollary \ref{th:cor_edge_dynamic})}
\label{fig:edge_algo}
\end{figure}

\begin{proof}
The correctness of the algorithm is a consequence of Theorem \ref{th:edge_dynamic} and the fact that the split tree transformations are safe (see Figure~\ref{fig:edge_algo}).
Let us turn to the complexity analysis. We assume (as we did in Section 4) that an artificial root of the split tree is maintained (remember that the graph and the split tree are connected).
Step 1 can be done easily in constant time, by searching the split tree in parallel from $x$ and $y$ towards the root (if the least common ancestor of $x$ and $y$ is found after $4$ steps or more, then the length of the path $P(x,y)$ is larger than $4$).
Step 2 also requires constant time.
There are at most two node-split operations and two node-join operations respectively at steps (a) and (c), each of which is constant time since it involves a ternary node.
And the transformation at step (b) is obviously constant time.
\end{proof}

\begin{figure}[htbh] 
\inserfig{0.95}{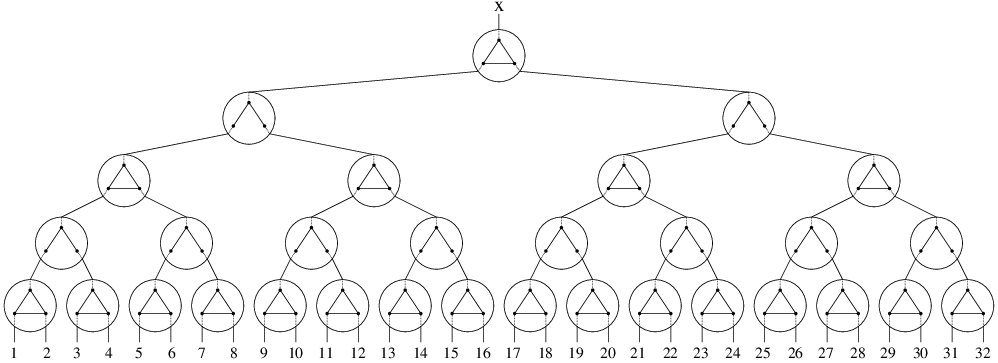}
%\centerline{\includegraphics[width=0.95\linewidth]{figures/ex_split_tree.eps}}
\caption{A DH graph (and cograph) such that removing any edge incident to the vertex $x$ provides a non-DH graph: the length of the path from $x$ to any other leaf is greater than 5.
}
\label{fig:contre-exemple}
\end{figure}

\begin{remark}
\label{rk:contre-exemple}
From Theorem \ref{th:edge_dynamic}, we can easily build an example of a DH graph (and cograph)
having a vertex such that removing {\bf any} edge incident to this vertex  provides a non-DH graph. It is depicted on Figure \ref{fig:contre-exemple}. This example shows that an edge-only dynamic recognition algorithm for DH graphs  cannot be used to obtain a vertex-only one.
\end{remark}
%\bigskip

%%%%%%%%%%%%%%%%%%%%%%%%%%%ù
\subsection{Proof of Theorem \ref{th:edge_dynamic}}

As mentioned above, our edge-modification characterization of DH graphs relies on the forbidden induced subgraph characterization: a graph is distance hereditary if and only if it does not contain a
cycle of length at least $5$ ($C_k$ for $k\geqslant 5$), a \emph{gem}, a \emph{house}, nor a \emph{domino}  (see Figure~\ref{fig:HDG}) as induced subgraph~\cite{BM86}. 

\begin{figure}[h]
\inserfig{0.7}{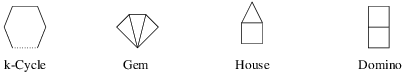}
\caption{The gem, the  house, and the domino  are together with the cycles $C_k$ ($k\geqslant 5$) the forbidden induced subgraphs for DH graphs}
\label{fig:HDG}
\end{figure}
We first need to introduce some notations and to state some basic properties and technical lemmas. 
We call \emph{factor}, in a word $W$, a set of consecutive letters of $W$.
We call \emph{$S$-subword}
a word obtained from $W$ by deleting some letters different from $S$. 
As for the clique-star trees, we say that a word is \emph{reduced} if it does not contain the following factors: $KK$, $S_yS_y$, $S_xS_x$, $S_yS$ and $SS_x$. 
With a word $W=w_1w_2\dots w_r$ on $A$, one can associate a clique-star tree $P_W$ whose underlying tree is a path of ternary nodes with hanging leaves
(i.e. $P_W$ is a {caterpillar}). 
Say that the first and last extreme nodes respectively have leaves $x$ and $y$, chosen to be the \emph{extreme-leaves} of $W$.
Then, the nodes of $P_W$ are labelled by graphs (with three vertices) accordingly to the letters of $W$ w.r.t. $x$ and $y$, just the same way as $P(x,y)$ corresponds to $W(x,y)$, as defined in the beginning of this section. 
We will denote $G_W$ the DH graph defined as the accessibility graph of the clique-star tree $P_W$.
Let $W$ be a word on $A$ with extreme-leaves $x,y$. Assuming $xy\not\in E(G_W)$, the word $W$  is called \emph{forbidden for edge-insertion} if $G_W+xy$ is not a DH graph; otherwise $W$ is \emph{safe for edge-insertion}. Simlarly, assuming $xy\in E(G_W)$, the word $W$   is  called \emph{forbidden for edge-deletion} if  $G_W-xy$ is not a DH graph; otherwise $W$ is \emph{safe for edge-deletion}.
The proof of Theorem \ref{th:edge_dynamic} relies on a characterization of the safe words (for insertion and deletion) by forbidden excluded subwords.

\begin{lemma} \label{lem:word_subgraph}
Let $x$ and $y$ be two vertices of a distance hereditary graph $G$. Then there exists a graph-labelled tree of $G$ with a node $u$ neighboring leaves $x$ and $y$ such that $G_u$ is isomorphic to $G_{W(x,y)}$.
Hence, in particular, $G_{W(x,y)}$ is isomorphic to an induced subgraph of $G$.
\end{lemma}

\begin{proof}
By definition, the graph-labelled tree $P_{W(x,y)}$ is isomorphic to the graph-labelled tree
obtained from $P(x,y)$ by substituting all nodes in $P(x,y)$ with ternary nodes corresponding to the same letters. Hence,
node-splitting in $ST(G)$ all nodes belonging to $P(x,y)$, in such a way that the path from $x$ to $y$ is preserved in the tree structure and is now labelled by ternary nodes, yields a subtree isomorphic to $P_{W(x,y)}$.
Joining all the nodes of this subtree provides a node, adjacent to leaves $x$ and $y$, and whose label is isomorphic to $G_{W(x,y)}$. It follows from Corollary~\ref{cor:subgraph} that $G_{W(x,y)}$ is an induced subgraph of $G$.
\end{proof}

\begin{lemma} \label{lem:edge_CNS}
Let $x$ and $y$ be two vertices of a distance hereditary graph $G=(V,E)$. 
If $xy\not\in E$, the graph $G+xy$ is distance hereditary if and only if the word $W(x,y)$ is not forbidden for edge-insertion.
If $xy\in E$, the graph $G-xy$ is distance hereditary if and only if the word $W(x,y)$ is not forbidden for edge-deletion.
\end{lemma}

\begin{proof}
Assume $xy\not\in E$. By definition, the graph $G_{W(x,y)}+xy$ is DH if and only if 
$W(x,y)$ is not forbidden for edge-insertion.
We prove that $G+xy$ is DH if and only if $G_{W(x,y)}+xy$ is DH.
By Lemma \ref{lem:word_subgraph}, there exists a graph-labelled tree $(T,\mathcal{F})$ of $G$ containing a node $u$ such that $G_u$ is isomorphic to $G_{W(x,y)}$.
As leaves $x$ and $y$ are adjacent to the node $u$ of $T$, replacing $G_u$ with $G_{W(x,y)}+xy$ yields a graph labelled tree whose accessibility graph is $G+xy$.
As a graph $G$ is DH if and only if all labels in a graph-labelled tree of $G$ are DH, the result obviously follows.
The proof for edge-deletion is similar.
\end{proof}

\begin{lemma} \label{lem:subgraph}
Let $x$ and $y$ be two vertices of a distance hereditary graph $G$. 
Every connected induced subgraph $H$ of $G_{W(x,y)}$ with $x,y\in V(H)$
is isomorphic to some graph $G_{W_H}$ where $W_H$ is a $S$-subword  of $W(x,y)$. Conversely, every such graph $G_{W_H}$ is isomorphic to some such connected induced subgraph $H$.
\end{lemma}

\begin{proof}
Let $W=W(x,y)=w_1...w_r$, and let $\{x=z_0,z_1,\dots z_r,y=z_{r+1}\}$
be the set of vertices of $G_W$, such that the ordering $z_1,...,z_r$ corresponds to the ordering of leaves encountered from $x$ to $y$ in the caterpillar $P_W$.
Let $H$ be an induced subgraph of $G_W$ such that $V(H)=\{x,z_{i_1}\dots z_{i_k},y\}$ with $i_1<\dots < i_k$.
Since $P_W$ is a caterpillar, $H$ is connected if and only if 
for every bipartition $(A,B)$ of $V(H)$, such that $A=\{z_i\in V(H)\mid i\leqslant j<r+1\}$ and $B=\{z_i\in V(H)\mid 0<j<i\}$ for some $j$, $H$ contains an edge between some vertex of $A$ and some vertex of $B$.
By the definition of accessibility, such an edge exists if and only if none of the letters $w_j$ of $W$ such that $z_j\not\in V(H)$ is a $S$. It follows that $H$ is connected if and only if the word $W_H=w_{i_1}w_{i_2}...w_{i_{k}}$ is a $S$-subword of $W$.
Finally, as an edge exists between two vertices of $G_W[V(H)]$
if and only if the corresponding letters in $W$ can be joined
by a sequence of letters in $\{K,S_x,S_y\}$, we have that $G_W[V(H)]$ is isomorphic to $G_{W_H}$. 
Also, the converse is straightforward.
\end{proof}
\bigskip

Let us consider the DH graphs obtained by removing, resp. adding, an edge $xy$ to one of the DH forbidden induced subgraphs $H$ (cycles, gem, house or domino). 
It turns out that the split tree of each one is a caterpillar with ternary nodes (see Figure \ref{fig:subwords_insertion}, resp. Figure \ref{fig:subwords_deletion}). 
Hence, they are determined by their associated words denoted $W_{H-xy}(x,y)$, resp. $W_{H+xy}(x,y)$.

\begin{figure} %[h]
\inserfig{0.85}{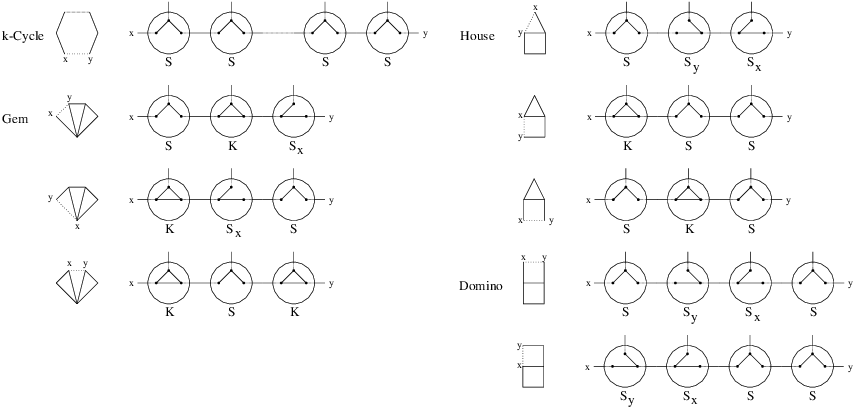}
%\centerline{\includegraphics[width=0.85\linewidth]{figures/subwords_insertion.pdf}}
\caption{proof of Theorem \ref{th:edge_dynamic}, insertion case. Split trees of graphs $H-xy$ for $H$ a DH forbidden induced subgraph.  Only useful graphs, i.e. DH ones, are represented. In the table of insertion forbidden subwords, in comparison, repetitions are deleted, and symmetric words are added.}
\label{fig:subwords_insertion}
\end{figure}

\begin{figure}[h]
\inserfig{0.85}{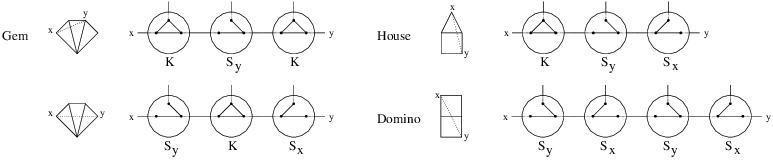}
\caption{proof of Theorem \ref{th:edge_dynamic}, deletion case. Split trees of graphs $H+xy$ for $H$ a DH forbidden induced subgraph.  Only useful graphs, i.e. DH ones, are represented. In the table of 
deletion forbidden subwords, in comparison, repetitions are deleted, and symmetric words are added.}
\label{fig:subwords_deletion}
\end{figure}

\begin{lemma} \label{lem:subwords}
A word $W$ with extreme-leaves $x,y$ is forbidden for edge-insertion, resp. edge-deletion,
if and only if it has a $S$-subword of type $W_{H-xy}(x,y)$, resp.
$W_{H+xy}(x,y)$, for $H$ a distance hereditary
forbidden induced subgraph.
\end{lemma}

\begin{proof}
We prove the statement for edge-insertion. Edge-deletion case is similar.
By definition, a word $W$, whose extreme-leaves are $x$ and $y$, is forbidden for edge-insertion if and only if $G_W+xy$ is not DH, i.e. $G_W+xy$ contains one of the DH forbidden induced subgraphs, say $H$, which also contains vertices $x$ and $y$ (since $G_W$ is DH). 
Now by Lemma \ref{lem:subgraph}, $H-xy$ is a connected induced subgraph of $G_W$ with vertices $x$ and $y$ if and only if
$H-xy=G_{W'}$ for some $S$-subword $W'$ of $W$,
that is if and only if the word $W'=W_{H-xy}(x,y)$ defined by the caterpillar split tree of $H-xy$
is a $S$-subword of $W$.
\end{proof}

\bigskip
For each  DH forbidden induced subgraph $H$ such that $H-xy$ (resp. $H+xy$) is DH, we obtain a list of \emph{(edge-)insertion forbidden subwords},
resp. \emph{(edge-)deletion forbidden subwords}, of type $W_{H-xy}(x,y)$, resp.
$W_{H+xy}(x,y)$. They are given by the following tables.

\begin{center}
\begin{tabular}{|c||c|c|c|c|}
\hline
subgraphs & $C_k$ ($k\leqslant 5$)& ~~~~Gem~~~~ & ~~~House~~~ & ~Domino~ \\
\hline
\hline
& & $SKS_x$ &  $SS_yR$ &  \\
insertion &  $S\ldots S$ & $S_yKS$ & $S_yS_xS$ &  $SS_yS_xS$  \\
forbidden  &  with & $KS_xS$ & $KSS$ &  $S_yS_xSS$ \\
subwords &$k\geqslant 3$ S's & $SS_yK$ & $SSK$ & $SSS_yS_x$ \\
& & $KSK$ & $SKS$ &  \\
\hline
\end{tabular}
\end{center}

\begin{center}
\begin{tabular}{|c||c|c|c|}
\hline
subgraphs & ~~~~Gem~~~~&  ~~~House~~~ & ~Domino~ \\
\hline
\hline
deletion & $KS_yK$ &  $KS_yS_x$ &     \\
forbidden  & $KS_xK$ &   $S_yS_xK$ &  $S_yS_xS_yS_x$  \\
subwords & $S_yKS_x$ &  &  \\
\hline
\end{tabular}
\end{center}

\bigskip

\noindent{\it Proof of Theorem \ref{th:edge_dynamic}:}
By Lemma \ref{lem:edge_CNS} and Lemma \ref{lem:subwords},
it remains to show that no $S$-subword of $W(x,y)$ belongs to the list of edge-insertion (resp. edge-deletion) forbidden subwords if and only if $W(x,y)$ is one of words described in condition 1 (resp. condition 2) of the theorem. Observe that the words of conditions 1 and 2 do not contain any forbidden words. Let us prove the converse.

\begin{enumerate}
\item %Insertion case. 
Assume that no $S$-subword of $W(x,y)$ belongs to the list of forbidden subwords for edge-insertion.

Notice that $W(x,y)$ contains at most two $S$'s otherwise it would contain a forbidden subwords corresponding to the cycles $C_k$, $k\geqslant 5$.

First consider the case $W(x,y)$ contains two $S$'s. By the $KSS$, $SSK$, $SKS$ House's forbidden subwords, $W(x,y)$ has no $K$ letter. By the Domino's forbidden subwords $W(x,y)$ is of the form $(S_x)SS(S_y)$. More precisely, as $W(x,y)$ is reduced, it does not contain the factors $S_xS_x$ or $S_yS_y$ and if a factor with no $S$ contains a $S_x$ (resp. a $S_y$), then $S_x$ (resp. $S_y)$ has to be the first (resp. last) letter of that factor. It follows that $W(x,y)$ is of the form $(S_x)(S_y)S(S_x)(S_y)S(S_x)(S_y)$. But again since $W(x,y)$ is reduced, it does not contain the factors $S_yS$ or $SS_x$. Thereby $W(x,y)$ is of the form $(S_x)SS(S_y)$.

Let us now consider the case $W(x,y)$ contains only one $S$. Hence $W(x,y)$ is of the form $wSw'$ where $w$ and $w'$ are reduced words on $\{K,S_y,S_x\}$. Then, by the $SS_yS_x$, $S_yS_xS$ House's forbidden subwords, $S_yS_x$ is not a $S$-subword of $w$ and $w'$. By the $SS_yK, S_yKS$ Gem's forbidden subword,  $KS_x$ is not a $S$-subword of $w$ and $w'$ neither. It follows that $w$ and $w'$ can only be the words $(S_x)(K)(S_y)$. More precisely if $w$ or $w'$ contains a $S_x$ (resp. $S_y$), then $S_x$ (resp. $S_y$) has to be the first (resp. last) letter. Moreover by the $KSK$ Gem's forbidden subword, at most one word among $w$ and $w'$ contains a $K$ letter. Finally, since $W(x,y)$ is reduced, it does not contain $S_yS$ or $SS_x$ as factors. Thereby $W(x,y)$ is of the form $(S_x)(K)S(S_y)$ or $(S_x)S(K)(S_y)$.

\item %Deletion case.
Assume that no $S$-subword $W(x,y)$ belongs to the list of forbidden subwords for edge-deletion.

First $W(x,y)$ contains at most one letter $K$.
Otherwise, since it is reduced and contains no factor $KK$, it would contain a $KS_yK$ or $KS_xK$ Gem's excluded subword.

Assume that $W(x,y)$ contains one letter $K$. 
Hence $W(x,y)$ is of the form $xKx'$ where
$x$ and $x'$ are reduced words on $\{S_y,S_x\}$.
Then, by the $KS_yS_x$, $S_yS_xK$ House's excluded subwords, 
$x$ and $x'$ must be of the form $(S_x)(S_y)$.
Precisely, if $x$ or $x'$ contains a $S_x$, resp. $S_y$, then $S_x$, resp. $S_y$, must be the first letter, resp. last. 
Moreover, by the  $S_yKS_x$ Gem's excluded subword, 
if $x$ contains a $S_y$, resp. $x'$ contains a $S_x$, 
then $x'$ does not contain a $S_x$ letter, resp. $x$ does not contain a $S_y$.
Hence $W(x,y)$ is of the form $(S_x)(S_y)K(S_y)$ or $(S_x)K(S_x)(S_y)$.

Assume $W(x,y)$ does not contains the letter $K$.
Then, by the $S_yS_xS_yS_x$ Domino's excluded subword, 
$W(x,y)$ must be of the form $(S_x)(S_y)(S_x)(S_y)$,
where any letter in  brackets can be deleted if it gives a reduced word, that is of the form $(S_x)S_yS_x(S_y)$ or $(S_x)(S_y)$.
\endproof
\end{enumerate}

%------------------------------------------------------------------------------------------------------------------
\subsection{Edge-modification in cographs.}
\label{sub:edge-cograph}

As already mentioned, cographs are $P_4$-free graphs (see Figure \ref{fig:p4}).
The split tree of a $P_4$ on vertices $\{x,y,a,b\}$ is formed by two adjacent star nodes,
hence it is associated with the word $W=SS$, $S_xS_y$, $S_xS$ or $SS_y$ depending on which leaves $x$ and $y$ correspond to.
Adapting Theorem \ref{th:edge_dynamic} and Corollary \ref{th:cor_edge_dynamic} leads to a similar characterization and a similar constant time algorithm, 
equivalent to the one given in \cite{SS04} in terms of cotrees.

The characterization and algorithm for connected cographs are obtained simply by replacing, in Theorem \ref{th:edge_dynamic} and Corollary \ref{th:cor_edge_dynamic},
the list of possible words $W(x,y)$ and respective transformations of the split tree, by the ones given in the following table.

\medskip
{
\noindent
\hfill{
\begin{tabular}{|c|c|}
\hline
\multicolumn{2}{| c |}{
edge-insertion $\longrightarrow$
}\\
\multicolumn{2}{| c |}{
$\longleftarrow$ edge-deletion
}\\
\hline

$\ \ \ SK\ \ \ $				& $S_yK$ \\ 
$KS$				& $KS_x$ \\ 
$S$				& $K$ \\ 

\hline
\end{tabular}
}\hfill
}

\medskip
Using the transformation linking the split decomposition to the modular decomposition, one can check that the above words correspond to the cotree configurations identified in~\cite{SS04} that allow edge-insertion or edge-deletion in a cograph.

\begin{theorem}
There exists an edge-modification fully dynamic recognition algorithm for connected cographs, maintaining  the split tree, with complexity $O(1)$ per edge-insertion or edge-deletion.
\end{theorem}

\begin{proof}
Let $G$ be a cograph and $x,y$ be two vertices of $G$. Since a cograph is DH, the necessary conditions of Theorem \ref{th:edge_dynamic} apply to $G$, that is: $W(x,y)$ belongs to the lists of words provided by Theorem \ref{th:edge_dynamic}.
Moreover the word $W(x,y)$ cannot contain $SS$, $S_xS$, $SS_y$ or $S_xS_y$ as a $S$-subword
(otherwise, by Lemma \ref{lem:accessible}, a $P_4$ would exist in $G$, which is a cograph: contradiction).
It is straightforward to check that the only words satisfying these properties
in the lists given in Theorem \ref{th:edge_dynamic} are those given in the actual theorem, namely:
\begin{enumerate}
\item if $xy\notin E$ and $G+xy$ is a cograph, then $W(x,y)$ is either $S$, $SK$ or $KS$;
\item if $xy\in E$ and $G-xy$ is a cograph, then $W(x,y)$ is either $K$, $S_yK$ or $KS_x$.
\end{enumerate}

The transformations of the edge-modification algorithm of cograph (see the above table) are special cases
of  the ones described and analysed in Corollary \ref{th:cor_edge_dynamic}. So the time complexity follows. It remains to check that  the transformations of the split tree, described in the above table, do not create a $P_4$ in the edge-modified graph.
\begin{itemize}
\item Assume that $W(x,y)=SK$ and that $G+xy$ has an induced $P_4$ (the case $W(x,y)=KS$ is symmetric). Let $\{a,b,x,y\}$ be the vertices of that $P_4$. As the split tree of $G$ is partitioned into the path $W(x,y)$ and two subtrees respectively attached to the $S$ node (resp. $K$ node) of $P(x,y)$ and  disjoint from $P(x,y)$, the vertices $a$ and $b$ cannot be leaves of the same subtree. Let $T_a$ be the subtree containing the leaf $a$ and $T_b$ be the subtree containing the leaf $b$. Let us note that since $W(x,y)=SK$,
the three vertices $y$, $a$ and $b$ induce a clique and none of its edges is modified in $G+xy$, contradicting the fact that $\{a,b,x,y\}$ induces a $P_4$ in $G+xy$.

\item Assume that $W(x,y)=S_yK$ and that $G-xy$ has an induced $P_4$ (the case $W(x,y)=KS_x$ is symmetric). As before let $\{a,b,x,y\}$ be the vertices of that $P_4$ and let $T_a$ (resp. $T_b$) be the maximal tree of $ST(G)-W(x,y)$ containing the leaf $a$ (resp. $b$).  This again implies that $\{a,b,y\}$ induces a clique that is not modified by the removal of $xy$, contradicting the fact that $\{a,b,x,y\}$ induces a $P_4$ in $G-xy$.

\item Assume $W(x,y)=S$ and thus $xy\notin E$ (resp. $W(x,y)=K$ and thus $xy\in E$), then $x$ and $y$ are false (resp. true) twins in $G$ ($xa$ is an edge if and only if $ya$ is an edge). This remains unchanged by the insertion (resp. deletion) of $xy$: there is no $P_4$ in $G+xy$ (resp. $G-xy$) containing $x$ and $y$: $G+xy$ (resp. $G-xy$) is a cograph.%
\end{itemize}%
\end{proof}

%------------------------------------------------------------------------------------------------------------------
\subsection{Edge-modification in $3$-leaf powers.}
\label{sub:edge-3leaf}

As 3-leaf power are DH, the edge-insertion must satisfy the properties of the edge-insertion in DH graph.
As a corollary of Theorem \ref{th:3-leaf}, the split tree of a 3-leaf power graph
does not contain a path of three nodes labelled successively by a star, a clique, and a  star.
Hence, the words  $SKS_x$, $S_yKS$, $S_xKS_x$, $S_yKS_y$ and $SKS$, 
have to be deleted from the list of Theorem \ref{th:edge_dynamic}.
This simplification turns out to be not sufficient,
conditions on the degrees and on adjacent nodes have to be added.

The characterization and algorithm for connected 3-leaf power graphs are obtained by:

\begin{enumerate} 

\item replacing, in Theorem \ref{th:edge_dynamic} and Corollary \ref{th:cor_edge_dynamic},
the list of possible words $W(x,y)$ and respective transformations of the split tree, by the ones given in the following table.
\smallskip

{
\noindent
\hfill{
\begin{tabular}{|c|c|}
\hline
\multicolumn{2}{| c |}{
edge-insertion $\longrightarrow$
}\\
\multicolumn{2}{| c |}{
$\longleftarrow$ edge-deletion
}\\
\hline

$(S_x)SK$				& $(S_x)S_yK$ \\ 
$KS(S_y)$				& $KS_x(S_y)$ \\ 
$(S_x)S$				& $(S_x)K$ \\ 
$S(S_y)$				& $K(S_y)$ \\ 
\hline
\end{tabular}
}\hfill
}

\item adding supplementary conditions on the safe words:

\begin{enumerate}

\item If $W(x,y)\in\{(S_x)SK, (S_x)S_yK,KS(S_y), KS_x(S_y)\}$, then
the letter  corresponding to a star and which is not in  brackets must come from a ternary star node $u$ of $ST(G)$ 
such that the neighbor of $u$ not in $P(x,y)$ is either a clique or a leaf.

\item If $W(x,y)\in\{S_xS,SS_y\}$, then
the node $u$ corresponding to letter $S$ must come from a ternary star node of $ST(G)$ such that the neighbor of $u$ not in $P(x,y)$ is either a clique or a leaf.

\item If $W(x,y)=K$, then
the corresponding clique node $u$ is either ternary and adjacent to a star node which is not oriented towards $u$,
or is not ternary, but the unique node of $ST(G)$.

\end{enumerate}
\end{enumerate}

\begin{theorem}
There exists an edge-modification fully-dynamic recognition algorithm for connected 3-leaf power graphs, maintaining  the split tree, with complexity $O(1)$ per edge-insertion or edge-deletion.
\end{theorem}

\begin{proof}
Since a 3-leaf power graph is DH, the necessary conditions of Theorem \ref{th:edge_dynamic} remain necessary for 3-leaf power graphs, 
that is: $W(x,y)$ is in the list of words provided by Theorem \ref{th:edge_dynamic}.
We recall that, by Theorem \ref{th:3-leaf}, the split tree of a DH graph is the split tree of a 3-leaf graph
if and only if the set of star nodes form a connected subtree and every star is oriented towards a clique or a leaf.
Hence, the word $W(x,y)$ cannot contain a letter $K$ between two letters corresponding to stars $S$, $S_x$, or $S_y$.
It is straightforward to check that the only words satisfying these properties
in the lists given in Theorem \ref{th:edge_dynamic} are those given in the above table plus
the associated words $(S_x)SS(S_y)$ and $(S_x)S_yS_x(S_y)$ (which are obtained from each other by respectively the insertion of $xy$ and the deletion of $xy$).
These latter two words cannot be considered in the list for 3-leaf power graphs,
since, in $(S_x)S_yS_x(S_y)$, two star nodes are oriented towards a star,
which would contradict Theorem \ref{th:3-leaf}.
The transformations provided by the actual algorithm are particular cases
of  the ones provided in Corollary \ref{th:cor_edge_dynamic}.

So it remains to check that the supplementary conditions on the degrees are necessary and sufficient
to have that the graph modified by these transformations is still a $3$-leaf power.

\begin{itemize}
\item Assume $W(x,y)=(S_x)SK$ is transformed into $(S_x)S_yK$ under the insertion of $xy$. 
Let $u$ be the node of $P(x,y)$ which gives the letter $S$ in  $W(x,y)$. If $u$ is not ternary, it has to be node-split into two star nodes $u'$ and $v$, with node $u'$ still belonging to $P(x,y)$ (see step 2.a in algorithm of Corollary~\ref{th:cor_edge_dynamic}). Then in $ST(G+xy)$, node $v$ is made adjacent to the a clique node of $P(x,y)$ (see step 2.b) in algorithm of Corollary~\ref{th:cor_edge_dynamic}). This is in contradiction with Theorem \ref{th:3-leaf} since that clique node neighbors two star nodes. 
Suppose now that the $S$ node $u$ is ternary in $ST(G)$ and let $T_a$ (resp. $T_b$) be the maximal tree of $ST(G)-P(x,y)$ attached to $u$ (resp. the $K$ node of $P(x,y)$). It follows that $ST(G+xy)$ satisfies the conditions of Theorem \ref{th:3-leaf} since $T_a$ is a clique or a leaf,
and the tree $T_b$ a leaf.

\item Assume $W(x,y)=(S_x)S_yK$ is transformed into $(S_x)SK$ under the deletion of $xy$. Let $u$ be the node of $P(x,y)$ which gives the letter $S_y$ in $W(x,y)$. As in the previous case, $u$ has to be a ternary node. Otherwise, it has to be node-split into two star nodes $u'$ and $v$, with $u'$ still belonging to $P(x,y)$. Again by the transformation algorithm described in Corollary~\ref{th:cor_edge_dynamic}, the clique node of $P(x,y)$ in $ST(G-xy)$ is neighboring two star nodes, contradicting Theorem \ref{th:3-leaf}. Finally,  by Theorem \ref{th:3-leaf}, the node of $ST(G)-P(x,y)$ adjacent to $u$ is a clique or a leaf and the node of $ST(G)-P(x,y)$ adjacent to the clique node of $P(x,y)$ is a leaf. It follows that $ST(G-xy)$ satisfies the conditions of Theorem \ref{th:3-leaf}.

\item The cases  $W(x,y)=KS(S_y)$ and $W(x,y)=KS_x(S_y)$ are symmetric to the previous ones.

\item Assume $W(x,y)=S_xS$ is transformed into $S_xK$ under the insertion of $xy$.  
The same arguments as above imply that the node giving letter $S$ is ternary (and hence has its centre adjacent to a clique or a leaf),
otherwise a clique would appear between two stars while inserting $xy$, contradicting Theorem \ref{th:3-leaf}.

\item Assume $W(x,y)=S_xK$ is transformed into $S_xS$ under the deletion of $xy$. Then
$ST(G-xy)$ necessary satisfies the conditions of Theorem \ref{th:3-leaf} since the clique node in $P(x,y)$ is adjacent to  a leaf by Theorem \ref{th:3-leaf} condition 2.

\item The cases  $W(x,y)=SS_y$ and $W(x,y)=KS_y$ are symmetric to the previous ones.

\item  Assume $W(x,y)=S$ is transformed into $K$ under the insertion of $xy$. Let $u$ be the node of $ST(G)$ which gives the letter $S$ in $W(x,y)$ and let $v$ be the neighbor of $u$ such that $\rho_u(uv)$ is the centre of the star $G_u$. By Theorem \ref{th:3-leaf}, $v$ is either a clique node of a leaf. If $u$ is a ternary node, then $G+xy$ is a clique and hence a $3$-leaf power. Otherwise, $u$ has to be node-split into two star nodes $u'$ and $v$ (as in the previous cases). In $ST(G+xy)$, the node $u'$ neighboring the leaves $x$ and $y$ is changed into a clique node. It follows that $ST(G+xy)$ satisfies the conditions of Theorem \ref{th:3-leaf}.

\item Assume $W(x,y)=K$ is transformed into $S$ under the deletion of $xy$.
If the clique node $u$ in $P(x,y)$ is not ternary, then all its neighbors in $ST(G)$ have to be leaves. Assume $u$ neighbors a star node $w$. As the edge-modification algorithm splits $u$ into two clique nodes $u'$ and $v$ and then change $u'$ into a star, the clique node $v$ would neighbor two star nodes in $ST(G-xy)$, contradicting Theorem \ref{th:3-leaf}.
So assume $u$ is ternary, then its third neighbor $v$ distinct from $x$ and $y$ is a leaf or a star. If $v$ is a star and $\rho_v(uv)$ is the centre of the star $G_v$, then the conditions of Theorem \ref{th:3-leaf} are not satisfied. Otherwise, it is clear that $ST(G-xy)$ satisfies  the conditions of Theorem \ref{th:3-leaf}.
\end{itemize}

Finally, we just have to check that the supplementary conditions can be tested in constant time.
The fact that a node of $P(x,y)$ is ternary or not can be checked in constant time.
When the node is ternary, the fact that the adjacent node not in $P(x,y)$ is a clique, or a leaf,
or a star oriented towards the clique, is also constant time
by checking the type of this adjacent node, and in the last case,
by testing whether the centre marker vertex of the star is mapped to a tree-edge incident to a ternary node.
In the last case where a clique is not ternary, testing if all other nodes of $ST(G)$ are leaves
is done simply by testing if $G$ is a clique.
\end{proof}

%------------------------------------------------------------------------------------------------------------------
%------------------------------------------------------------------------------------------------------------------

\bibliographystyle{plain}

%%%%%%%%%%%%%%%%%%%%%%%%%%%%%%%%%%%%%%%%%%%%%%%%ù
%%%%%%%%%%%%%%%%%%%%%%%%%%%%%%%%%%%%%%%%%%%%%%%%%

%\newpage
\appendix

\section{Appendix}

%%%%%%%%%%%%%%%%%%%%%%%%%%%%%
% EM: PREUVE DE TH AJOUTEE AVEC PROP de rencesement des splits
% inutile pour modular decomp car consequence de ce th et de la particularisation de modular decomp
%%%%%%%%%%%%%%%%%%%%%%%%%%%

For self-containment %completeness 
of the paper, we provide below a direct proof of Theorem \ref{th:cun}
in the setting of graph-labelled trees. % which is slightly more precise than in \cite{Cun82}.
It relies %directly 
on next Proposition \ref{prop:split_list}, a result already underlying in \cite{Cun82}, somehow the converse of Lemma \ref{lem:split-edge}, 
and providing also a %direct 
proof 
%
%Note that Proposition \ref{prop:split_list} is closely related
to the well-known fact that the splits of the graph form a bipartitive~family. 
%[???], and thus splits forming `strong' elements of this family can be represented as edges of tree, which turns out to be exactly the split tree.
%However, for completeness again, and in order to give an explicit construction, we give a direct proof of this result.

\begin{proposition} \label{prop:split_list}
Let $(T,\mathcal{F})$ be a reduced graph-labelled tree with prime and degenerate labels obtained by split decomposition of a connected graph $G=(V,E)$.
Then every split of the graph $G$ is the bipartition induced by removing a tree-edge
of $T'$, where $T'$ is obtained from $(T,\mathcal{F})$ by at most one node-split of a degenerate node.
%
%Then any split of the graph $G$ is the bipartition induced by an edge
%of $T$, up to splitting a degenerate node.
\end{proposition}

\begin{proof}
Let $(A,B)$  be a split of $G$, with $V=A\uplus B$,
$A'\subseteq A$, $B'\subseteq B$, and all edges between $A$ and $B$
having their extremities in $A'$ and $B'$.
We consider $V$ as the set of leaves of $T$.
For a node $N$ of $T$ and a vertex $v$ of the label $G_N\in\mathcal{F}$ of $N$,
we say that $v$ is $A'$-accessible, resp. $B'$-accessible, 
if there exists $u\in A'$, resp. $u\in B'$, such that $N$ is $u$-accessible and
$v$ is the marker-vertex of $G_N$ associated with the tree-edge of $T$ in the path from $N$ to $u$.

First, we prove the following assertion: if $N$ is a node of $T$ with label $G_N$
and three  vertices $u,v,w$ of $G_N$ are such that $u$ is both $A'$-accessible and $B'$-accessible,
$v\not= u$ is $A'$-accessible and $w\not= u$ is $B'$-accessible, then $N$ is a clique and every vertex of $G_N$ is either $A'$-accessible or $B'$-accessible.
We can assume $v\not=w$. Indeed, if $v=w$, then let $x$ be a vertex of $G_N$ adjacent to $u$ or $v=w$ (it exists by connectivity Lemma \ref{lem:connected}).
By Lemma \ref{lem:accessible}, there exists a leaf $x'$ of $T$ such that $N$ is $x'$-accessible
and $x$ is associated with the tree-edge of $T$ in the path between $N$ and $x'$ in $T$.
Then $x'$ is adjacent to a vertex in $A'$ and a vertex in $B'$, hence it belongs to $A'\cup B'$,
hence $x$ is $A'$-accessible or $B'$-accessible. So we can change $v$ or $w$ into $x$, and we assume now that $v\not=w$ in the above hypothesis.
Since $V=A\uplus B$ is a split of $G$, there is an edge in $G_N$ 
between every $A'$-accessible vertex of $G_N$ and every $B'$-accessible vertex of $G_N$.
Assume there exist a vertex $y$ of $G_N$ which is not $A'$-accessible nor $B'$-accessible.
By Lemma \ref{lem:accessible}, there exists a leaf $z$ of $T$ such that $N$ is $z$-accessible
and $y$ is associated with the tree-edge of $T$ in the path between $N$ and $z$ in $T$.
Such a leaf $z$ belongs either to $A\setminus A'$ or to $B\setminus B'$, otherwise there would be an edge
between $y$ and $u$ in $G_N$ and $z$ would be adjacent to a vertex in $A'$ and to a vertex in $B'$, hence
it would belong to $A'\cup B'$ and $y$ would be $A'$-accessible or $B'$-accessible.
Assume for example that $z$ belongs to $A\setminus A'$. Then the vertex $y$ cannot be adjacent to a $B'$-accessible vertex in $G_N$.
Then we consider the bipartition of vertices of $G_N$ into 
1) all $A'$-accessible vertices except $u$, plus the vertex $y$, plus all other vertices 
which are not $A'$-accessible nor $B'$-accessible and not adjacent to a $B'$-accessible vertex in $G_N$,
and 2) all $B'$-accessible vertices including $u$, plus other vertices 
which are not $A'$-accessible nor $B'$-accessible and not adjacent to a $A'$-accessible vertex in $G_N$.
The two parts of this bipartition have at least two elements, and thus it forms a split of $G_N$.
This implies $G_N$ is degenerate, hence it is a clique since it contains a $K_3$.
And this implies that a vertex $y$ which is not $A'$-accessible nor $B'$-accessible does not exist.

Now consider the subtrees $T[A']$ and $T[B']$ of $T$ spanned respectively by $A'$ and $B'$.
We prove the assertion: $T[A']$ and $T[B']$ have at most one common node.
Assume that these two subtrees have at least two common nodes of $T$. Then they have a common path of $T$, and then
they have a common tree-edge $e$ and two adjacent common nodes $R$ and $S$, which are the extremities of $e$.
For each of these nodes, by definition of $T[A']$ and $T[B']$, 
there exists a leaf in $A'$ and a leaf in $B'$ such that the path from the node to the leaf does not contain $e$.
Since all leaves in $A'$ are accessible from all leaves in $B'$, each node $R$ and $S$ has an $A'$-accessible vertex and a $B'$-accessible vertex,
not associated with $e$. That is: both $R$ and $S$ satisfy the hypothesis of the previous assertion.
Hence both $R$ and $S$ are cliques, a contradiction with the fact that the graph-labelled tree is reduced.

So, there are two possible cases:  either 1) there exists  a tree-edge $e$ of $T$ such that all leaves in $A'$ are in one connected component of $T-e$,
and all leaves in $B'$ are in the other connected component, 
or 2) there exists a node $N$ such that the connected components of $T-N$ contain either leaves in $A'$ or leaves in $B'$ but not both, and at least two connected components contain leaves in $A'$ resp. $B'$ (otherwise a tree-edge would satisfy the case 1 property). 

In the first case, we denote $N_A$, resp. $N_B$ the extremity of the tree-edge $e$ on the side of $A'$, resp. $B'$, vertices. 
If the bipartition of $V$ induced by $T-e$ is $V=A\uplus B$, then we have the result.
Otherwise, there exists for example a leaf $a$ in $A$ in the connected component of $T-e$ containing leaves in $B'$.
Consider the accessibility graph $H$ of this connected component, where the marker-vertex associated with $e$ in $N_B$ is a vertex $v_e$.
Since $H$ does not contain vertices in $A'$, a vertex in $A$ is adjacent in $H$ either to $v_e$ or to another vertex belonging to $A$.
Hence $v_e$ is an articulation vertex of $H$, such that each connected component of $H-v_e$ has its set of vertices included in $A\setminus A'$ or in $B$.
Then the vertex $v_e$ of $N_B$ is an articulation vertex of its label. 
Indeed, say that a vertex $v\not=v_e$ is $A$-accessible, resp. $B$-accessible, 
if there exists a leaf $w\in A$, resp. $w\in B$, such that $N_B$ is $w$-accessible and $v$ is associated to the path between $N_B$ and $w$.
Then the vertices of $N_B-v_e$ are either $A$-accessible or $B$-accessible, maybe both, but these two sets are not empty and
a $A$-accessible vertex is not adjacent to a $B$-accessible vertex.
Since the label of $N_B$ has an articulation vertex, it has obviously a split, hence it is degenerate, and it is not a clique, hence it is a star.
Then a vertex $v$ of $N_B-v_e$ is either $A$-accessible or $B$-accessible, but not both. 
Otherwise the node at the other extremity of the tree-edge associated with $v$ would have the same property as $N_B$ in terms of articulation vertex, 
hence it would be a star also, which would be a contradiction with the fact that the graph-labelled tree is reduced.
Now, there cannot be a leaf $b$ in $B$ in the connected component of $T-e$ containing leaves in $A'$, otherwise
$N_A$ would also be a star and there would be an edge between a vertex of $A\setminus A'$ and a vertex of $B\setminus B'$.
So, finally, splitting the node $N_B$ into $v_e$ plus the $A$-accessible vertices on one hand, and the $B$-accessible vertices on the other hand,
creates a tree-edge $e$ which separates the leaves of the tree into $A$ and $B$.

In the second case, since there is an edge in $G$ between every vertex in $A'$ and  every vertex in $B'$,
then every marker-vertex in the label of $N$ is either $A'$-accessible or $B'$-accessible, but not both.
Hence there is a split in $N$ formed by $A'$-accessible and $B'$-accessible vertices. 
Hence $N$ is a clique node. Assume a leaf $a\in A$ is in a connected component of $T-N$ containing vertices in $B'$. Then $a\not\in A'$.
Let $v$ be the marker-vertex of $N$ associated with the tree-edge of $T$ adjacent to this connected component.
Then $v$ has to be $A$-accessible. Otherwise, the accessibility graph of the connected component would not be connected, contradicting Lemma \ref{lem:connected}.
Since there are at least two connected components of $T-N$ containing vertices in $B'$, and since $N$ is a clique,
then there is an edge between a vertex in $A\setminus A'$ and a vertex in $B'$, which is a contradiction.
So the connected components of $T-N$ contain either leaves in $A$ or leaves in $B$ but not both.
And, finally, splitting the node $N$ respectively with this partition
creates a tree-edge $e$ which separates the leaves of the tree into $A$ and $B$.
\end{proof}
\bigskip

\noindent {\it Proof of Theorem \ref{th:cun}}:
With Proposition \ref{prop:split_list}, the list of splits of the graph determines the list of partitions of leaves of the tree induced by tree-edge or node removals.
It is an easy combinatorial property that this list determines completely the tree. Then the labels are necessarily determined by the graph structure.
So, the whole graph-labelled tree is uniquely determined by the graph.
\endproof

\end{document}